\documentclass[screen]{acmart}

\usepackage{booktabs}   

\usepackage{amsmath}
\usepackage{amsthm}
\usepackage{stmaryrd}

\usepackage{amssymb}
\usepackage{enumerate}
\usepackage{todonotes}
\usepackage{bussproofs}
\usepackage{mathtools}
\usepackage{makecell}
\usepackage{comment}
\usepackage{proof}
\usepackage[strings]{underscore}
\usepackage[misc]{ifsym}
\usepackage{enumitem}
\allowdisplaybreaks

\newcommand{\adcomment}[1]{}

\newcommand{\revII}[1]{#1}

\newtheorem{remark}{Remark}

\newcommand{\N}{\mathbb{N}}

\newcommand{\I}{\mathcal{I}}
\newcommand{\R}{\mathcal{R}}

\newcommand{\Fm}{F_\textsf{m}}
\newcommand{\fs}{\sigma_{\textsf{f}}}
\newcommand{\sfs}{\sigma_{\textsf{S(f)}}}
\newcommand{\ite}{\mathit{ite}}
\newcommand{\nil}{\mathit{nil}}
\newcommand{\Mod}{\textsf{Mod}}
\newcommand{\lei}{\le_{\mathrm{i}}}
\newcommand{\slei}{<_{\mathrm{i}}}
\newcommand{\mui}{\mu_{\mathrm{i}}}

\newcommand{\lef}{\le_{\mathrm{f}}}
\newcommand{\slef}{<_{\mathrm{f}}}
\newcommand{\muf}{\mu_{\mathrm{f}}}

\newcommand{\dummyloc}{\textit{kf}}
\newcommand{\Dummyloc}{\textit{KF}}
\newcommand{\wtp}{\mathit{wtp}}
\newcommand{\mwptr}[2]{\mathit{MW}^{#1.f:=#2}}
\newcommand{\mwalloc}[1]{\mathit{MW}_v^{\mathsf{alloc}(#1)}}
\newcommand{\mwallocparam}[2]{\mathit{MW}_{#1}^{\mathsf{alloc}(#2)}}
\newcommand{\halloc}[1]{\Fr^{x;v}}

\newcommand{\sem}[2][M]{\llbracket #2 \rrbracket_{#1}}
\newcommand{\semnu}[2][M,\nu]{\llbracket #2 \rrbracket_{#1}}

\newcommand{\Fr}{\textit{Sp}}
\newcommand{\mh}{\textrm{PSL }}

\newcommand{\valid}{\mathit{valid}}

\newcommand{\ttree}{\mathit{ttree}}
\newcommand{\pttree}{\mathit{pttree}}
\newcommand{\lft}{\mathit{left}}
\newcommand{\rght}{\mathit{right}}

\newcommand{\tnext}{\mathit{tnext}}

\newcommand{\bst}{\mathit{bst}}
\newcommand{\key}{\mathit{key}}
\newcommand{\maxkey}{\mathit{maxkey}}
\newcommand{\minkey}{\mathit{minkey}}

\newcommand{\btree}{\mathit{btree}}
\newcommand{\lst}{\mathit{list}}
\newcommand{\lseg}{\mathit{lseg}}
\newcommand{\nxt}{\mathit{next}}
\newcommand{\dll}{\mathit{dll}}
\newcommand{\prv}{\mathit{prev}}

\newcommand{\len}{\mathit{length}}
\newcommand{\mkee}{\mathit{mkeys}}

\newcommand{\slst}{\mathit{slist}}
\newcommand{\treap}{\mathit{treap}}
\newcommand{\pty}{\mathit{priority}}
\newcommand{\ptset}{\mathit{priorities}}
\newcommand{\bfac}{\mathit{bfac}}
\newcommand{\height}{\mathit{height}}
\newcommand{\avl}{\mathit{avl}}

\makeatletter
\newtheorem*{rep@theorem}{\rep@title}
\newcommand{\newreptheorem}[2]{%
\newenvironment{rep#1}[1]{%
 \def\rep@title{#2 \ref{##1}}%
 \begin{rep@theorem}}%
 {\end{rep@theorem}}}
\makeatother

\newreptheorem{theorem}{Theorem}

\makeatletter
\newtheorem*{rep@lemma}{\rep@title}
\newcommand{\newreplemma}[2]{%
\newenvironment{rep#1}[1]{%
 \begin{rep@lemma}}%
 {\end{rep@lemma}}}
\makeatother

\newreptheorem{lemma}{Lemma}

\newcommand{\spp}{support}
\newcommand{\spps}{supports}
\newcommand{\Spp}{Support}
\newcommand{\Spps}{Supports}

\newcommand{\Star}{\mathit{Star}}

\makeatletter
\DeclareRobustCommand{\vcup}[1]{%
  \mathbin{\vphantom{\cup}\mathpalette\vcup@{#1}}%
}
\newcommand{\vcup@}[2]{%
  \ooalign{%
    $\m@th#1\cup$\cr
    \hidewidth
    \raisebox{%
      \dimexpr\depth+\demote@{#1}\dimexpr1.5pt\relax\relax
    }{\fontsize{\demote@{#1}\dimexpr\f@size pt}{\z@}\selectfont\check@mathfonts
      $\m@th#2$}%
    \hidewidth\cr
  }%
}
\newcommand{\demote@}[1]{%
  \ifx#1\displaystyle 0.7
  \else\ifx#1\textstyle 0.7
  \else\ifx#1\scriptstyle 0.5
  \else 0.4\fi\fi\fi
}
\makeatother

\makeatletter
\newcommand{\printfnsymbol}[1]{%
  \textsuperscript{\@fnsymbol{#1}}%
}
\makeatother

\AtBeginDocument{%
  }

\setcopyright{acmcopyright}
\copyrightyear{2018}
\acmYear{2018}
\acmDOI{XXXXXXX.XXXXXXX}

\acmConference[Conference acronym 'XX]{Make sure to enter the correct
  conference title from your rights confirmation emai}{June 03--05,
  2018}{Woodstock, NY}
\acmPrice{15.00}
\acmISBN{978-1-4503-XXXX-X/18/06}

\begin{document}

\title{A First-Order Logic with Frames}

\author{Adithya Murali}
\email{adithya5@illinois.edu}
\affiliation{
  \institution{University of Illinois at Urbana-Champaign, Department of Computer Science}
  \city{Urbana}
  \state{Illinois}
  \country{USA}
}

\author{Lucas Pe\~{n}a}
\email{lpena7@illinois.edu}
\affiliation{
  \institution{University of Illinois at Urbana-Champaign, Department of Computer Science}
  \city{Urbana}
  \state{Illinois}
  \country{USA}
}

\author{Christof L\"{o}ding}
\email{loeding@automata.rwth-aachen.de}
\affiliation{
  \institution{RWTH Aachen University, Department of Computer Science}
  \city{Aachen}
  \country{Germany}
}

\author{P. Madhusudan}
\email{madhu@illinois.edu}
\affiliation{
  \institution{University of Illinois at Urbana-Champaign, Department of Computer Science}
  \city{Urbana}
  \state{Illinois}
  \country{USA}
}

\renewcommand{\shortauthors}{Adithya Murali et al.}

\begin{abstract}
We propose a novel logic, called \emph{Frame Logic} (FL), that extends first-order logic (with recursive definitions) using a construct $\Fr(\cdot)$ that captures
the \emph{implicit \spps{}} of formulas--- the precise subset of the universe upon which their meaning depends. Using such \spps{}, we formulate proof rules that facilitate frame reasoning elegantly when the underlying model undergoes change. 
We show that the logic is expressive by capturing several data-structures and also exhibit a translation from a \emph{precise} fragment of separation logic to frame logic. 
Finally, we design a program logic based on frame logic for reasoning with programs
that dynamically update heaps that facilitates local specifications and frame reasoning. This program logic
consists of both localized proof rules as well as rules that derive the weakest tightest preconditions in FL.
\end{abstract}



\maketitle

\section{Introduction}\label{sec:intro}

Program logics for expressing and reasoning with programs that dynamically manipulate heaps is an active area of research. The research on separation logic has argued convincingly that it is highly desirable to have \emph{localized logics} that talk about small states (heaplets rather than the global heap), and the ability to do \emph{frame reasoning}. Separation logic achieves this objective by having a tight heaplet semantics and using special operators, primarily a separating conjunction operator $*$ and a separating implication operator (the magic wand $-*$).

In this paper, we ask a fundamental question: can logics such as FOL and FOL with recursive definitions be extended to support localized specifications and frame reasoning?
Can we utilize such logics for reasoning effectively with programs that dynamically manipulate heaps, with the aid of local specifications and frame reasoning?

The primary contribution of this paper is to endow the first-order logic with recursive definitions (with least fixpoint semantics) with frames and frame reasoning.

A formula in first-order logic with recursive definitions (FO-RD) can be naturally associated with a \emph{\spp{}}--- the subset of the universe that determines its truth. By using a more careful syntax such as guarded quantification (which continue to have a classical interpretation), we can in fact write specifications in FO-RD that have very precise \spps{}.
For example, we can write the property that $x$ points to a linked list using a formula $list(x)$ written purely in FO-RD so that its \spp{} is precisely the locations constituting the linked list.

In this paper, we define an extension of FO-RD, called
Frame Logic (FL) where we allow a new operator $\Fr(\alpha)$ which, for an FO-RD formula $\alpha$, evaluates to the \spp{} of $\alpha$. Logical formulas
thus have access to \spps{} and can use it to \emph{separate} \spps{} and do frame reasoning. For instance, the logic can now express that two lists are disjoint by asserting that
$\Fr(list(x)) \cap \Fr(list(y)) = \emptyset$. It can then reason that in such a program heap configuration, if the program
manipulates only the locations in $\Fr(list(y))$, then $list(x)$
would continue to be true, using simple frame reasoning.
We can also avoid blowup while doing this by defining macros, such as $\textit{Star}(\alpha, \beta) = \alpha \wedge \beta \wedge \Fr(\alpha) \cap \Fr(\beta)=\emptyset$. 
In addition, FL allows for us to omit this subformula expressing disjoint supports. We can simply write the formula $\lst(x) \land \lst(y)$ to represent lists that \emph{may or may not} be overlapping. Such formulas are very difficult to write in Separation Logic.

The addition of the \spp{} operator to FO-RD yields a very natural logic for expressing specifications. First, formulas in FO-RD have the same meaning when viewed as FL formulas. For example,
$f(x)=y$ (written in FO-RD as well as in FL) is true in any model that has $x$ mapped by $f$ to $y$, instead of a specialized ``tight heaplet semantics'' that demands that $f$ be a partial function with the domain only consisting of the location $x$. The fact that the \spp{} of this formula contains only the location $x$ is important, of course, but is made accessible using the \spp{} operator, i.e., $\Fr(f(x)=y)$ gives the singleton set containing the element interpreted for $x$\footnote{\revII{For a pointer $f$ we denote by $f(x)=y$ that $x$ points to $y$ on $f$ i.e.,   $x \overset{f}{\mapsto} y$, whose ``footprint'' is $\{x\}$ in many heap logics}}. 
Second, properties of \spps{}
can be naturally expressed using set operations. 
To state that the lists pointed to by $x$ and $y$ are disjoint, we don't need special operators (such as the $*$ operator in separation logic) but can express this as $\Fr(list(x)) \cap \Fr(list(y))=\emptyset$. 
Third, when used to annotate programs, pre/post specifications for programs written in FL can be made \emph{implicitly} local by interpreting their \spps{} to be the localized heaplets accessed and modified by programs, yielding frame reasoning akin to program logics that use separation logic. Finally, as we show in this paper, the weakest precondition of specifications across basic loop-free paths can be expressed in FL, making it an expressive logic for reasoning with programs. Separation logic, on the other hand, introduces the magic wand operator $-*$ (which is inherently higher-order) in order to add enough expressiveness to be closed under weakest preconditions~\cite{reynolds02}. 

We define frame logic (FL) as an extension of FO with recursive definitions (FO-RD) that operates over a multi-sorted universe, with a particular foreground sort (used to model locations on the heap on which pointers can mutate) and several background sorts that are defined using separate theories. \Spps{} for formulas are defined with respect to the foreground sort only. A special background sort of \emph{sets} of elements of the foreground sort is assumed and is used to model the \spps{} for formulas. For any formula $\varphi$ in the logic, we have a special construct $\Fr(\varphi)$ that captures its \spp{}, a set of locations in the foreground sort, that intuitively corresponds to the precise subdomain of functions the value of $\varphi$ depends on. We then prove a \emph{frame theorem} (Theorem~\ref{thm:frametheorem}) that says that changing a model $M$ by changing the interpretation of functions that are not in the \spp{} of $\varphi$ will not affect the truth of the formula $\varphi$.
This theorem then directly supports frame reasoning; if a  model satisfies $\varphi$ and the model is changed so that the changes made are disjoint from the \spp{} of $\varphi$, then $\varphi$ will continue to hold. We also show that FL formulas can be translated to vanilla FO-RD logic (without \spp{} operators); in other words, the semantics for the \spp{} of a formula can be captured in FO-RD itself\footnote{\revII{The blowup of translation from FL to FO-RD is only linear; this can be seen from Figure~\ref{fig:frame-equations}.}}. Consequently, we can use any FO-RD reasoning mechanism (proof systems~\cite{vampire,kovacs17} or heuristic algorithms such as the natural proof techniques~\cite{madhusudan12,pek14,qiu13,suter10}) to reason with FL formulas.

We illustrate our logic using several examples drawn from program verification; we show how to express various
data-structure definitions and the elements they contain 
and various measures for them using FL formulas (e.g., linked lists, sorted lists, list segments, binary search trees, AVL trees, lengths of lists, heights of trees, set of keys stored in the data-structure, etc.)

While the sensibilities of our logic are definitely inspired by separation logic, there are some fundamental 
differences beyond the fact that our logic extends the syntax and semantics of first-order logic with recursive definitions with a special \spp{} operator and avoids operators such as $*$ and $-*$. In separation logic, there can be many \spps{} of a formula (also called heaplets)--- a heaplet for a formula is one that \emph{supports its truth}. For example, a formula of the form $\alpha \vee \beta$ can have a heaplet that supports the truth of $\alpha$ or one that supports the truth of $\beta$. However, the philosophy that we follow in our design is to have a \emph{single} \spp{} that supports the truth value of a formula, whether it be \emph{true or false}. Consequently, the \spp{} of the formula 
$\alpha \vee \beta$ is the \emph{union} of the \spps{} of the formulas $\alpha$ and~$\beta$. 

The above design choice of the \spp{} being \emph{determined} by the formula has several consequences that lead to a deviation from separation logic. For instance, the \spp{} of the negation of a formula $\varphi$ is the same as the \spp{} of $\varphi$. And the \spp{} of the formula $f(x)=y$ and its negation are the same, namely the singleton location interpreted for $x$.
In separation logic, the corresponding formula will have the same heaplet but its negation will include \emph{all} other heaplets. The choice of having determined \spps{} or heaplets is not new, and there have been several variants and sublogics of separation logics that have been explored. For example, the logic {\sc Dryad}~\cite{qiu13,pek14} is a separation logic that insists on determined heaplets to support automated reasoning, and the \emph{precise} fragment of separation logic studied in the literature~\cite{ohearn04} defines a sublogic that has (essentially) determined heaplets. The second main contribution in this paper is to show that this fragment of separation logic (with slight changes for technical reasons) can be translated to frame logic, such that the unique heaplet that satisfies a precise separation logic formula is its \spp{} of the corresponding formula in frame logic.

The third main contribution of this paper is a program logic based on frame logic for a simple while-programming language destructively updating heaps. We present two kinds of proof rules for reasoning with such programs annotated with pre- and post-conditions written in frame logic. The first set of rules are local rules that axiomatically define the semantics of the program, using the smallest \spps{} for each command. We also give a frame rule that allows arguing preservation of properties whose \spps{} are disjoint from the heaplet modified by a program. These rules are similar to analogous rules in separation logic. The second class of rules work to give a \emph{weakest tightest precondition} for any postcondition with respect to non-recursive programs. In separation logic, the corresponding rules for weakest preconditions are often expressed using separating implication (the magic wand operator). Given a small change made to the heap and a postcondition $\beta$, the formula $\alpha$ $-*$ $\beta$ captures all heaplets $H$ where if a heaplet that satisfies $\alpha$ is joined with $H$, then $\beta$ holds.
When $\alpha$ describes the change effected by the program, $\alpha$ $-*$ $\beta$ captures, essentially, the weakest precondition. However, the magic wand is a very powerful operator that calls for quantifications over heaplets and submodels, and hence involves second-order quantification~\cite{magic-wand-complexity}. In our logic, we show that we can capture the weakest precondition with only first-order quantification, and hence first-order frame logic is closed under weakest preconditions across non-recursive programs blocks. This  means that when inductive loop invariants are given also in FL, reasoning with programs reduces to reasoning with FL. By translating FL to pure FO-RD formulas, we can use FO-RD reasoning techniques to reason with FL, and hence programs.

\revII{Our work is in large part inspired by our previous work on reasoning with \textsc{Dryad}~\cite{qiu13,pek14,loding18} using FO-reasoning engines, in particular SMT solvers. {\sc Dryad} is a separation logic that has, essentially, \emph{unique heaplets}, and this is crucial in translation to FO without quantifiers over sets (for example, quantifying over two heaplets to reason with $\alpha * \beta$). {\sc Dryad} does not have the magic wand and is therefore not closed under weakest preconditions. We imbibe this philosophy fundamentally in the design of Frame Logic--- the support for formulas is uniquely defined. When converting Dryad formulas to FO-RD, the translation has separate components for capturing truth and heaplets (especially heaplets for inductively defined datastructures), and the correctness of the heaplet definition is not derived or verified to be correct in the works in~\cite{qiu13,pek14,loding18}. Frame Logic addresses this disadvantage by ensuring that the support of formulas are systematically captured using first-order logic (by the Frame Logic to FO-RD translation). Frame Logic is also closed under weakest preconditions. 
Reasoning with {\sc Dryad} also proceeds using a verification condition technique based on strongest post, which avoids introducing unnecessary quantification. From our experience in working with both {\sc Dryad} and Frame Logic, we believe that logical engines to reason with FO-RD formulas derived from Frame Logic will be as successfully amenable to automation as {\sc Dryad} is, especially when adapted to verification conditions derived using the strongest-post. Such an extension is considerable effort, however, and we leave this to future work. We provide further discussion in Section~\ref{sec:reasoning-with-fl}.}

\smallskip

\noindent In summary, the contributions of this paper are:
\begin{itemize}[topsep=0pt]
 \item A logic, called \emph{frame logic} (FL) that extends FO-RD with a \spp{} operator and supports frame reasoning. We illustrate FL with specifications of
 various data-structures.
 We show a translation to equivalent formulas in 
 FO-RD.
 \item A program logic and proof system based on FL including local rules and rules for computing the weakest tightest precondition. FL reasoning
 required for proving programs is hence reducible
 to reasoning with FO-RD.
 \item A separation logic fragment that can generate only precise formulas, and a translation from this logic to equivalent FL formulas.
\end{itemize}

The paper is organized as follows.
Section~\ref{sec:ford} sets up first-order logics with recursive definitions (FO-RD), with a special uninterpreted foreground sort of locations and several background sorts/theories. Section~\ref{sec:frame-logic} introduces Frame Logic (FL), its syntax, its semantics which includes a discussion of design choices for supports, proves the frame theorem for FL, shows a reduction of FL to FO-RD, and illustrates the logic by defining several data-structures and their properties using FL. 
Section~\ref{sec:prog-and-prfs} develops a program logic based on FL, illustrating them with proofs of verification of programs. Section~\ref{sec:seplog-translation} introduces a precise fragment of separation logic and shows its translation to FL. Section~\ref{sec:discussion} discusses comparisons of FL to separation logic, and some existing first-order  techniques that can be used to reason with FL. Section~\ref{sec:rel-work} compares our work with the research literature and Section~\ref{sec:conc} has concluding remarks.

\revII{This paper is an extension of our work at ESOP 2020~\cite{murali20}. In this work we present a cleaner operational semantics, a simpler program logic, and give detailed proofs of all our results. We also present detailed comparisons in the discussion section (Sec.~\ref{sec:discussion}) and the section on related work (Sec.~\ref{sec:rel-work}).}

\section{Background: First-Order Logic with Recursive Definitions
and Uninterpreted Combinations of Theories} 
\label{sec:ford}

The base logic upon which we build frame logic is a first 
order logic with recursive definitions (FO-RD), where we allow a foreground
sort and several background
sorts, each with their individual theories (like arithmetic, sets, arrays, etc.).
\revII{The logic FO-RD is essentially the same as the classical first-order logic with least fixpoints (FO+lfp) studied in finite model theory and databases from the 1980s~\cite{libkin04,vardi1982,immerman1982,relational-logic+lfp,chandra-harel}. The only difference is
that we give names to recursive definitions that have least fixpoint semantics. (Logics with inductive definitions over non-mononotonic bodies are also studied in the literature; see Section~\ref{sec:rel-work}) on related work.}

The foreground sort and functions involving the foreground sort are 
\emph{uninterpreted} (not constrained by theories). This hence can be seen as
an uninterpreted combination of theories over disjoint domains. 
This logic has been defined and used to model heap verification before~\cite{loding18}.

We will build frame logic over such a framework where \spps{} are modeled as 
subsets of elements of the foreground sort. When modeling heaps in program
verification using logic, the foreground sort will be used to model 
\emph{locations of the heap}, uninterpreted functions from the foreground sort to 
foreground sort will be used to model \emph{pointers}, and uninterpreted functions
from the foreground sort to the background sort will model \emph{data fields}.
Consequently, \spps{} will be subsets of locations of the heap, which is
appropriate as these are the domains of pointers that change when a program
updates a heap.

We define a signature as $\Sigma = (S; C; F; \R; \I)$, where
$S$ is a finite non-empty set of sorts.
$C$ is a set of constant symbols, where each $c \in C$ has some sort $\tau \in S$.
$F$ is a set of function symbols, where each function $f \in F$ has a type of the
form $\tau_1 \times \ldots \times \tau_m \rightarrow \tau$ for some $m$, with
$\tau_i, \tau \in S$. 
The sets $\R$ and $\I$ are (disjoint) sets of relation symbols, where each relation $R \in \R \cup \I$ has
a type of the form $\tau_1 \times \ldots \times \tau_m$. The set $\I$
contains those relation symbols for which the corresponding relations
are inductively defined using formulas (details are given below), while those
in $\R$ are given by the model.

We assume that the set of sorts contains a designated ``foreground
sort'' denoted by $\fs$. 
All the other sorts in $S$ are called background sorts, and for
each such background sort $\tau$ we allow the constant symbols of type $\tau$, function symbols that have type
$\tau^n \rightarrow \tau$ for some $n$, and relation symbols have
type $\tau^m$ for some $m$, to be constrained using an arbitrary
theory $T_\tau$.

\newcommand{\recdef}{\mathcal{D}}
A formula in first-order logic with recursive definitions (FO-RD) over
such a signature is of the form $(\recdef, \alpha)$, where $\recdef$
is a set of recursive definitions of the form $R(\vec x) := \rho_R(\vec{x})$,
where $R \in \I$ and $\rho_R(\vec{x})$ is a first-order logic formula, in which the 
relation symbols from $\I$ occur only positively, where each symbol from $\I$ is under an even number of negations. $\alpha$ is also a first-order logic formula over the signature.
We assume $\recdef$ has at most one definition for any inductively
defined relation, and that the formulas $\rho_R$ and $\alpha$ use only inductive relations defined in $\recdef$. 

The semantics of a formula is standard; the semantics of inductively defined
relations are defined to be the least fixpoint that satisfies the relational
equations, and the semantics of $\alpha$ is the standard one defined using these
semantics for relations. We do not formally define the semantics, but we will formally define the semantics of frame logic which is an extension of FO-RD.

\section{Frame Logic} \label{sec:frame-logic}
We now define Frame Logic (FL), the central contribution of this paper.

\begin{figure}
\small
 \[
\begin{array}{rrcl}
  \text{FL formulas:} & \varphi &::=& \bot \mid \top \mid t_\tau = t_\tau \mid R(t_{\tau_1}, \ldots, t_{\tau_m}) \mid \varphi \land \varphi \mid \lnot \varphi \mid \ite(\gamma:\varphi,\varphi) \mid \exists y:\gamma. ~\varphi \\
  &&& \mbox{$\tau \in S$, $R \in \R \cup \I$ of type $\tau_1 \times \cdots \times \tau_m$} \\
 \text{Guards:} & \gamma &::=& t_\tau = t_\tau \mid R(t_{\tau_1}, \ldots, t_{\tau_m}) \mid \gamma \land \gamma \mid \lnot \gamma \mid  \ite(\gamma:\gamma,\gamma) \mid \exists y:\gamma. ~\gamma \\
  &&& \mbox{$\tau \in S \setminus \{\sfs\}$, $R \in \R$ of type $\tau_1 \times \cdots \times \tau_m$} \\
 \text{Terms:} & t_\tau &::=& c \mid x \mid f(t_{\tau_1}, \ldots, t_{\tau_m}) \mid \ite(\gamma:t_\tau,t_\tau) \mid \\ 
  &&& \Fr(\varphi) ~~\text{(if $\tau = \sfs$)} \mid \Fr(t_{\tau'}) ~~\text{(if $\tau = \sfs$)} \\
  &&& \mbox{$\tau,\tau' \in S$ with constants $c$, variables $x$ of type $\tau$,} \\
  &&& \mbox{and functions $f$ of type $\tau_1 \times \cdots \times t_m \rightarrow \tau$} \\
\multicolumn{4}{l}{\text{Recursive definitions: $R(\vec x) := \rho_R(\vec{x})$ with $R \in \I$ of type $\tau_1 \times \cdots \times \tau_m$ with}} \\
&&& \text{$\tau_i \in S \setminus\{\sfs\}$, FL formula $\rho_R(\vec{x})$ where all relation symbols} \\
&&& \text{ $R' \in \I$ occur only positively or inside a \spp{} expression.}
\end{array}
\normalsize
\]

\caption{Syntax of frame logic: $\gamma$ for guards, $t_\tau$ for terms of sort $\tau$, 
and general formulas $\varphi$. Guards cannot use inductively defined relations or \spp{} expressions.} \label{fig:syntax-frame-logic}
\end{figure}

We consider a universe with a foreground sort and several background sorts, each restricted by individual theories,
as described in Section~\ref{sec:ford}. 
We consider the elements of the foreground sort to be \emph{locations} 
and consider \spps{} as \emph{sets of locations}, i.e., sets of elements
of the foreground sort. 
We hence introduce a background sort $\sfs$; the elements of sort $\sfs$ model sets of elements of sort $\fs$. 
Among the relation symbols in $\R$ there is the relation
$\in$ of type $\fs \times \sfs$ that is interpreted as the usual
element relation. The signature includes the standard operations on
sets $\cup$, $\cap$ with the usual
meaning, the unary function $\widetilde{\cdot}$ that is
interpreted as the complement on sets (with respect to the set of
foreground elements), and the constant
$\emptyset$. For these functions and relations we assume a background theory $B_{\sfs}$
that is an axiomatization of the theory of
sets. We further assume that the signature does not contain any other
function or relation symbols involving the sort $\sfs$. 



For reasoning about changes of the structure over the locations, we assume that there is a nonempty
subset $\Fm \subseteq F$ of function symbols that are declared
mutable. These functions can be used to model mutable pointer fields in the heap
that can be manipulated by a program and thus change. Formally, we
require that each $f \in \Fm$ has at least one argument of sort
$\fs$.


For variables, let $\mathit{Var}_\tau$ denote the set of variables of
sort $\tau$, where $\tau \in S$.  We let $\overline{x}$ abbreviate
tuples $x_1, \ldots, x_n$ of variables.


Our frame logic over uninterpreted combinations of theories is a
variant of first-order logic with recursive definitions that has an additional operator
$\Fr(\varphi)$ that assigns to each formula $\varphi$ a set of elements
(its \spp{} or ``heaplet'' in the context of heaps) in the foreground universe. So $\Fr(\varphi)$ is a term of
sort $\sfs$.

The intended semantics of $\Fr(\varphi)$ (and of the
inductive relations) is defined formally 
as a least fixpoint of a set of
equations. This semantics is presented in
Section~\ref{sec:frames}. In the following, we first define the
syntax of the logic, then discuss informally the various design decisions for the semantics of \spps{}, before proceeding to a formal definition of the semantics

%

\subsection{Syntax of Frame Logic (FL)}
The syntax of our logic is given in the grammar in Figure~\ref{fig:syntax-frame-logic}. This extends FO-RD with the rule for building \emph{\spp{} expressions}, which are terms of sort $\sfs$ of the form $\Fr(\alpha)$ for a formula $\alpha$, or $\Fr(t)$ for a term $t$.


The formulas defined by $\gamma$ are used as \emph{guards} in existential quantification and in the if-then-else-operator, which is denoted by $\ite$. The restriction compared to general formulas is that guards cannot use inductively defined relations ($R$ ranges only over $\R$ in the rule for $\gamma$, and over $\R \cup \I$ in the rule for $\varphi$), nor terms of sort $\sfs$ and thus no \spp{} expressions ($\tau$ ranges over $S \setminus\{\sfs\}$ in the rules for $\gamma$ 
and over $S$ in the rule for $\varphi$). The requirement that the guard does not use the inductive
 relations and \spp{} expressions is used later to ensure the
 existence of least fixpoints for defining semantics of inductive
 definitions. The semantics of an $\ite$-formula $\ite(\gamma:\alpha,\beta)$ is the
 same as the one of $(\gamma \land \alpha) \lor (\lnot \gamma
 \land \beta)$; however, the \emph{\spps{}} of the two
 formulas will turn out to be different (i.e.,
 $\Fr(\ite(\gamma:\alpha,\beta))$ and $\Fr((\gamma \land
 \alpha) \lor (\lnot \gamma \land \beta))$ are
 different), as explained in Section~\ref{sec:design}.
The same is true for existential formulas, i.e., $\exists y:\gamma. \varphi$ has the same semantics as $\exists y.\gamma \land \varphi$ but, in general, has a different \spp{}.

As a brief example to help understand these syntactic constraints, consider the expression $\exists y : \lseg(x, y).\; \lst(y)$. This is not syntactically valid in Frame Logic as the guard contains the recursive definition $\lseg$. However, we could instead write the formula $\exists y : y = \nxt(x).\; \lst(y)$, which is syntactically valid as the guard only references variables and the mutable function $\nxt(x)$. 

\revII{Universal formulas of the kind $\forall x:\gamma.\,\alpha$ are shorthands for the formula $\neg(\exists x:\gamma.\,\neg\alpha)$ as expected, and their supports are hence defined as well.}

For recursive definitions (throughout the paper, we use the terms recursive definitions and inductive definitions with the same meaning), we require that the relation $R$ that is defined does not have arguments of sort $\sfs$. This is another restriction in order to ensure the existence of a least fixpoint model in the definition of the semantics.\footnote{It would be sufficient to restrict formulas of the form $R(t_1, \ldots, t_n)$ for inductive relations $R$ to not contain \spp{} expressions as subterms.}

\subsection{Semantics of \Spp{} Expressions: Design Decisions}\label{sec:design}
We discuss the design decisions that go behind the semantics of the \spp{} operator $\Fr$ in our logic, and then give an example for the \spp{} of an inductive definition. The formal conditions that the \spps{} should satisfy are stated in the equations in 
Figure~\ref{fig:frame-equations}, and are explained in Section~\ref{sec:frames}. Here, we start by an informal discussion.

The first decision is to have every formula uniquely define a \spp{}, which roughly captures the subdomain of mutable functions that a formula $\varphi$'s truthhood depends on, 
and have $\Fr(\varphi)$ evaluate to it.

The choice for \spps{} of atomic formulas are relatively clear. An atomic formula of the kind
$f(x)\!\!=\!\!y$, where $x$ is of the foreground sort and $f$ is a mutable function, has as its \spp{} the singleton set containing the location interpreted for $x$. And atomic formulas that do not involve mutable functions over the foreground have an empty \spp{}. \Spps{} for terms can also be similarly defined. The \spp{} of a conjunction $\alpha \wedge \beta$ should clearly be the union of the \spps{} of the two formulas.

\begin{remark}\label{rem:mutation-frame}
In traditional separation logic, each pointer field is stored in a separate location, using integer offsets.
However, in our work, we view pointers as references and disallow pointer arithmetic.
A more accurate heaplet for such references can be obtained by taking
heaplet to be the pair $(x,f)$ (see~\cite{Parkinson05}), capturing
the fact that the formula depends only on the
field $f$ of $x$.
Such accurate heaplets can be captured in FL as well--- we can introduce a \emph{non-mutable field lookup pointer} $L_f$
and use $x.L_f.f$ in programs instead of $x.f$.
\end{remark}


What should the \spp{} of a formula $\alpha \vee \beta$ be? The choice we make here is that its \spp{} is the \emph{union} of the \spps{} of $\alpha$ and $\beta$. Note that in a model where $\alpha$ is true and $\beta$ is false, we still include the heaplet of $\beta$ in $\Fr(\alpha \vee \beta)$. In a sense, this is an overapproximation of the \spp{} as far as frame reasoning goes, as surely preserving the model's definitions on the \spp{} of $\alpha$ will preserve 
the truth of $\alpha$, and hence of $\alpha \vee \beta$. 

However, we prefer the \spp{} to be the union of the \spps{} of $\alpha$ and $\beta$. We think of the \spp{} as the subdomain of the universe that determines the meaning of the formula, whether it be \emph{true} or \emph{false}. Consequently, we would like
the \spp{} of a formula and its negation to be the
same. Given that the \spp{} of the negation of a disjunction, being a conjunction, is the union of the frames of $\alpha$ and $\beta$, we would like this to be the \spp{}. 

Separation logic makes a different design decision.
Logical formulas are not associated with tight \spps{}, but rather, the semantics of the formula is defined for models with given \spps{}/heaplets, 
where the idea of a heaplet is whether it supports
the \emph{truthhood} of a formula (and not its falsehood).
For example, for a model, the various heaplets that satisfy $\neg (f(x)=y)$ in separation logic would include all heaplets where the location of $x$ is not present, which does not coincide with the notion we have chosen for \spps{}.
However, for positive formulas, separation logic handles \spps{} more accurately, as it can associate several \spps{} for a formula, yielding two heaplets for formulas of the form $\alpha \vee \beta$ when they are both true in a model.
The decision to have a single \spp{} for a formula compels us to take the union of the \spps{} to be the \spp{} of a disjunction.

There are situations, however, where there are disjunctions $\alpha \vee \beta$, where only \emph{one} of the disjuncts can possibly be true, and hence we would like the \spp{} of the formula to be the \spp{} of the disjunct that happens to be true. 
We therefore introduce a new syntactical form $ite(\gamma: \alpha, \beta)$ in frame logic, whose heaplet is the union of the \spps{} of $\gamma$ and $\alpha$, if $\gamma$ is true, and the \spps{} of $\gamma$ and $\beta$ if $\gamma$ is false. While the truthhood of 
$ite(\gamma: \alpha, \beta)$ is the same as
that of $(\gamma \wedge \alpha) \vee (\neg \gamma \wedge \beta)$, its \spps{} are potentially smaller,
allowing us to write formulas with tighter \spps{}
to support better frame reasoning. Note that the
\spp{} of $ite(\gamma : \alpha, \beta)$ and its negation $ite(\gamma : \neg \alpha, \neg \beta)$ are the same, as we desired.

Turning to quantification, the \spp{} for a formula of the form $\exists x. \alpha$ is hard to define, as its truthhood could depend on the entire universe.
We hence provide a mechanism for \emph{guarded}
quantification, in the form $\exists x: \gamma. ~~\alpha$. The semantics of this formula is
that there exists some location that satisfies the guard $\gamma$, for which $\alpha$ holds. 
The \spp{} for such a formula includes
the \spp{} of the guard, and the \spps{} of $\alpha$
when $x$ is interpreted to be a location that satisfies $\gamma$. For example, $\exists x: (x=f(y)).~~g(x)=z$ has as its \spp{} the locations interpreted for $y$ and $f(y)$ only, since we only consider the support of $g(x)$ when $x$ is interpreted to be $f(y)$.

For a formula $R(\overline{t})$ with an inductive relation $R$ defined by $R(\overline{x}) := \rho_R(\overline{x})$, the 
\spp{} descends into the definition, changing the variable assignment of the variables in $\overline{x}$ from the inductive definition to the terms in $\overline{t}$.  Furthermore, it contains the elements to which mutable functions are applied in the terms in $\overline{t}$.


Recursive definitions are designed such
that the evaluation of the equations for the \spp{} expressions is independent of
the interpretation of the inductive relations. The equations mainly depend on the
syntactic structure of formulas and terms. Only the semantics of guards, and the semantics of subterms under a mutable function symbol play a role. For this reason,
we disallow guards to contain recursively defined relations or \spp{} expressions. We also require that the only functions involving the sort $\sfs$ are the standard 
functions involving sets. Thus, subterms of 
mutable functions cannot contain \spp{} expressions (which are of sort $\sfs$) as subterms.

These restrictions ensure that there
indeed exists a unique simultaneous least solution of the equations for
the inductive relations and the \spp{} expressions. We now provide an example.

\begin{example} 
\label{ex:tree}
Consider the definition of a predicate $\mathit{tree}(x)$ w.r.t.\ two unary mutable functions $\lft$ and $\rght$:
\begin{align*}
  \mathit{tree}(x) := & \ite(x = \nil: \top, \alpha) \mbox{   where}\\
  \alpha = &\exists \ell,r:(\ell=\lft(x) \land r = \rght(x)). \mathit{tree}(\ell) \land \mathit{tree}(r) \land {} \\
  &\Fr(\mathit{tree}(\ell)) \cap \Fr(\mathit{tree}(r)) = \emptyset \land
  \neg ( x \in \Fr(\mathit{tree}(\ell)) \cup \Fr(\mathit{tree}(r)) )
\end{align*}
This inductive definition defines binary trees with pointer fields $\lft$ and $\rght$ for left- and right-pointers, by stating that $x$ points to a tree if either $x$ is equal to $nil$ (in this case its \spp{} is empty), or
$\lft(x)$ and $\rght(x)$ are trees with disjoint \spps{}.
The last conjunct says that $x$ does not belong to the \spp{} of the left and right subtrees; this condition is, strictly speaking, not required to define trees (under least fixpoint
semantics). 
Note that the access to the \spp{} of formulas eases defining disjointness of heaplets, like in separation logic.
The \spp{} of $tree(x)$ turns out to be precisely the nodes that are reachable from $x$ using $\lft$ and $\rght$ pointers, as one would desire. Consequently, if a pointer outside this \spp{} changes, we would be able to conclude using frame reasoning that the truth value of $tree(x)$ does not change.
\qed
\end{example}

\subsection{Formal Semantics of Frame Logic}\label{sec:frames}

\begin{figure}
\small
  \[
  \begin{array}{rcl}
    \sem{\Fr(c)}(\nu) &=& \sem{\Fr(x)}(\nu) = \emptyset \mbox{ for a constant $c$ or variable $x$}\\
  \sem{\Fr(f(t_1, \ldots, t_n))}(\nu) &=&
  \begin{cases}
  \bigcup\limits_{i \text{ with } t_i \text{ of sort } \fs}\{\semnu{t_i}\} \cup \bigcup\limits_{i=1}^n \sem{\Fr(t_i)}(\nu) & \mbox{if } f \in \Fm \\
  \bigcup\limits_{i=1}^n \sem{\Fr(t_i)}(\nu) & \mbox{if } f \not\in \Fm 
  \end{cases}\\
  \sem{\Fr(\Fr(\varphi))}(\nu) &=& \sem{\Fr(\varphi)}(\nu) \\
  \sem{\Fr(\Fr(t))}(\nu) &=& \sem{\Fr(t)}(\nu) \\
 \revII{\sem{\Fr(\top)}(\nu)} &\revII{=}& \revII{\sem{\Fr(\top)}(\nu) = \emptyset} \\
 \sem{\Fr(\bot)}(\nu) &=& \sem{\Fr(\top)}(\nu) = \emptyset \\
 \sem{\Fr(t_1 = t_2)}(\nu) &=& \sem{\Fr(t_1)}(\nu) \cup \sem{\Fr(t_2)}(\nu) \\
 \sem{\Fr(R(t_1, \ldots, t_n))}(\nu) &=& \bigcup_{i=1}^n \sem{\Fr(t_i)}(\nu)  \mbox{
  for $R \in \R$}\\
\sem{\Fr(R(\overline{t}))}(\nu) &=& \sem{\Fr(\rho_R(\overline{x}))}(\nu[\overline{x} \leftarrow \semnu{\overline{t}}]) \cup \bigcup_{i=1}^n \sem{\Fr(t_i)}(\nu)  \\
  && \text{for $R \in \I$ with definition $R(\overline{x}) := \rho_R(\overline{x})$}, \\
  && \overline{t} = (t_1, \ldots, t_n), \overline{x} = (x_1, \ldots, x_n)\\
\sem{\Fr(\alpha \land \beta)}(\nu) &=& \sem{\Fr(\alpha)}(\nu) \cup
\sem{\Fr(\beta)}(\nu) \\
\sem{\Fr(\lnot \varphi)}(\nu) &=& \sem{\Fr(\varphi)}(\nu) \\
\sem{\Fr(\ite(\gamma:\alpha,\beta))}(\nu)  &=& \sem{\Fr(\gamma)}(\nu) \cup
  \begin{cases}
     \sem{\Fr(\alpha)}(\nu) \mbox{ if } M,\nu \models \gamma \\
    \sem{\Fr(\beta)}(\nu) \mbox{ if } M,\nu \not\models \gamma 
  \end{cases} \\
\sem{\Fr(\ite(\gamma:t_1,t_2))}(\nu)  &=& \sem{\Fr(\gamma)}(\nu) \cup
  \begin{cases}
     \sem{\Fr(t_1)}(\nu) \mbox{ if } M,\nu \models \gamma \\
    \sem{\Fr(t_2)}(\nu) \mbox{ if } M,\nu \not\models \gamma 
  \end{cases} \\
   \sem{\Fr(\exists y: \gamma.\varphi)}(\nu) &=& \bigcup\limits_{u \in D_y} \sem{\Fr(\gamma)}(\nu[y \leftarrow u]) \cup \bigcup\limits_{u \in D_y; M,\nu[y
      \leftarrow u] \models \gamma} \sem{\Fr(\varphi)}(\nu[y \leftarrow u])
  \end{array}
  \]
  \caption{Equations for \spp{} expressions}
    \label{fig:frame-equations}
\normalsize
\end{figure}

Before we explain the semantics of the \spp{} expressions and inductive definitions, we introduce a semantics that treats \spp{} expressions and the symbols from $\I$ as uninterpreted symbols. We refer to this semantics as \emph{uninterpreted semantics}. For the formal definition we need to introduce some terminology first.

An occurrence of a variable $x$ in a formula is free if it does not
occur under the scope of a quantifier for $x$. By renaming 
variables we can assume that each variable only occurs freely in a
formula or is quantified by exactly one quantifier in the formula.  We
write $\varphi(x_1, \ldots, x_k)$ to indicate that the free variables
of $\varphi$ are among $x_1, \ldots, x_k$.  Substitution of a term $t$
for all free occurrences of variable $x$ in a formula $\varphi$ is
denoted $\varphi[t/x]$.  Multiple variables are substituted
simultaneously as $\varphi[t_1/x_1, \ldots, t_n/x_n]$.  We abbreviate
this by $\varphi[\overline{t}/\overline{x}]$.



A model is of the form $M = (U ; \sem\cdot)$ where $U =
(U_{\tau})_{\tau \in S}$ contains a universe for each sort, and an
interpretation function $\sem\cdot$. The universe for the sort $\sfs$
is the powerset of the universe for $\fs$. 

A variable assignment is a function $\nu$ that assigns to each
variable a concrete element from the universe for the sort of the
variable. For a variable $x$, we write $D_x$ for the universe of the
sort of $x$ (the domain of $x$). For a variable $x$ and an element $u
\in D_x$ we write $\nu[x \leftarrow u]$ for the variable assignment
that is obtained from $\nu$ by changing the value assigned for $x$ to
$u$.

The interpretation function $\sem\cdot$ maps each constant $c$ of sort $\tau$ to
an element $\sem{c} \in U_{\tau}$, each function symbol $f:\tau_1
\times \ldots \times \tau_m \rightarrow \tau$ to a concrete function
$\sem f : U_{\tau_1} \times \ldots \times U_{\tau_m} \rightarrow
U_\tau$, and each relation symbol $R \in \R \cup \I$ of type $\tau_1
\times \ldots \times \tau_m$ to a concrete relation $\sem R \subseteq
U_{\tau_1} \times \ldots \times U_{\tau_m}$.  These interpretations are assumed 
to satisfy the background theories (see Section~\ref{sec:ford}).
Furthermore, the
interpretation function maps each expression of the form $\Fr(\varphi)$
to a function $\sem{\Fr(\varphi)}$ that assigns to each variable
assignment $\nu$ a set $\sem{\Fr(\varphi)}(\nu)$ of foreground elements. The set
$\sem{\Fr(\varphi)}(\nu)$ corresponds to the \spp{} of the formula when
the free variables are interpreted by $\nu$.  Similarly,
$\sem{\Fr(t)}$ is a function from variable assignments to sets of foreground elements. 

Based on such models, we can define the semantics of terms and
formulas in the standard way.  The only construct that is non-standard
in our logic are terms of the form $\Fr(\varphi)$, for which the
semantics is directly given by the interpretation function in Figure~\ref{fig:frame-equations}. We write $\semnu{t}$ for the interpretation of a term $t$ in
$M$ with variable assignment $\nu$.  With this convention,
$\sem{\Fr(\varphi)}(\nu)$ denotes the same thing as
$\semnu{\Fr(\varphi)}$. As usual, we write $M,\nu \models \varphi$
to indicate that the formula $\varphi$ is true in $M$ with the free
variables interpreted by $\nu$, and $\sem{\varphi}$ denotes the relation
defined by the formula $\varphi$ with free variables $\overline{x}$. 
\revII{Note that the support of $\forall y:\gamma.\,\varphi$ is the same as the support of $\exists y:\gamma.\,\varphi$ since $\forall y:\gamma.\,\varphi$ is a shorthand for $\neg(\exists y:\gamma.\,\neg\varphi)$.}

We refer to the above semantics as the
\emph{uninterpreted semantics} of $\varphi$ because we do not give a
specific meaning to inductive definitions and \spp{} expressions.

Now let us define the true semantics for FL.
The relation symbols $R \in \I$ represent inductively defined
relations, which are defined by equations of the form $R(\overline{x}) :=
\rho_R(\overline{x})$ (see Figure~\ref{fig:syntax-frame-logic}).
In the intended meaning, $R$ is interpreted as the least relation that
satisfies the equation 
\[
\sem{R(\overline{x})} = \sem{\rho_R(\overline{x})}.
\]
The usual requirement for the existence of a unique least fixpoint of the equation is that the definition of $R$ does not
negatively depend on $R$. For this reason, we require that in
$\rho_R(\overline{x})$ each occurrence of an inductive predicate $R' \in \I$
is either inside a \spp{} expression, or it occurs under an even
number of negations.\footnote{As usual, it would be sufficient to
  forbid negative occurrences of inductive predicates in mutual
  recursion.} 


Every \spp{} expression is evaluated on a model to a set of foreground elements (under a given variable assignment $\nu$). 
Formally, we are interested in models in which the \spp{} expressions are interpreted to be the sets that correspond to the \emph{smallest solution of the equations given in Figure~\ref{fig:frame-equations}}. The intuition behind these definitions 
was explained in Section~\ref{sec:design}

\begin{example} \label{ex:tree-frame}
Consider the inductive definition $tree(x)$ defined in Example~\ref{ex:tree}.
%
To check whether the equations from Figure~\ref{fig:frame-equations} indeed yield the desired \spp{}, note that the \spps{} of $\Fr(x=nil) = \Fr(x) = \Fr(\top) = \emptyset$. Below, we write $[u]$ for a variable assignment that assigns $u$ to the free variable of the formula that we are considering. Then we obtain that $\Fr(\mathit{tree}(x))[u] = \emptyset$ if $u = \nil$, and $\Fr(\mathit{tree}(x))[u] = \Fr(\alpha)[u]$ if $x \not= nil$. The formula $\alpha$ is existentially quantified with guard $\ell = \lft(x) \land r = \rght(x)$. The {\spp{}} of this guard is $\{u\}$ because mutable functions are applied to $x$. 
The {\spp{}} of the remaining part of $\alpha$ is the union of the {\spps{}} of $\mathit{tree}(\ell)[\lft(u)]$ and 
$\mathit{tree}(r)[\rght(u)]$ (the assignments for $\ell$ and $r$ that make the guard true).
So we obtain for the case that $u \not = \nil$ that 
the element $u$ enters the \spp{}, and the recursion further descends into the subtrees of $u$, as desired.\qed
\end{example}

A \emph{frame model} is a model in which the interpretation of the inductive relations and of the \spp{} expressions 
corresponds to the least solution of the respective equations. \revII{We formalize this idea and prove the following proposition in Section~\ref{sec:frame-models}.}

\begin{proposition}\label{pro:unique-frame-model}
For each model $M$, there is a unique frame model over the same universe and the same interpretation of the constants, functions, and non-inductive relations.
\end{proposition}

\subsection{Frame Models}\label{sec:heaplet-models}\label{sec:frame-models}


In this section we present a formal construction of \emph{frame models} over which FL formulae are interpreted. We first introduce some useful terminology.


A \emph{pre-model} $\hat{M}$ is defined like a model with the
difference that a pre-model does not interpret the inductive relation
symbols and the \spp{} expressions $\Fr(\varphi)$ and $\Fr(t)$. A
pre-model $\hat{M}$ spans a class of models $\Mod(\hat{M})$, namely
those that simply extend $\hat{M}$ by an interpretation of the
inductive relations and the \spp{} expressions. 

The inductive definitions of relations from $\I$ can have 
negative references to \spp{} expressions. For example, the tree definition 
from Example~\ref{ex:tree} uses \spp{} expressions in the subformula 
$\Fr(\mathit{tree}(\ell(x))) \cap \Fr(\mathit{tree}(r(x))) = \emptyset$. This formula
is true if there does not exist an element in the intersection $\Fr(\mathit{tree}(\ell(x))) \cap \Fr(\mathit{tree}(r(x)))$, 
and hence negatively refers to these \spp{} expressions.
For this reason, we need to define two partial orders that correspond to first taking the least fixpoint for the \spp{} expression, and then the least fixpoint for the inductive predicates. 

In the following, we refer to the equations for the \spp{} expressions from Figure~\ref{fig:frame-equations} as \emph{\spp{} equations}, and to the equations $\sem{R(\overline{x})} = \sem{\rho_R(\overline{x})}$ for the inductive definitions as 
the \emph{inductive equations}.

For $M_1,M_2 \in \Mod(\hat{M})$ we let $M_1 \lef M_2$ if 
\begin{itemize}
\item $\sem[M_1]{\Fr(\varphi)}(\nu) \subseteq
\sem[M_2]{\Fr(\varphi)}(\nu)$ as well as $\sem[M_1]{\Fr(t)}(\nu) \subseteq
\sem[M_2]{\Fr(t)}(\nu)$ for all \spp{} expressions and all variable assignments $\nu$.
\end{itemize}
Note that $\lef$ is not a partial order but only a preorder: for two models $M_1,M_2$ that differ only in their interpretations of the inductive relations, we have $M_1 \lef M_2$ and $M_2 \lef M_1$. We write $M_1 \slef M_2$ if $M_1 \lef M_2$ and not $M_2 \lef M_1$.

We further define $M_1 \lei M_2$ if 
\begin{itemize}
\item $\sem[M_1]{\Fr(\varphi)} =
\sem[M_2]{\Fr(\varphi)}$ as well as $\sem[M_1]{\Fr(t)} =
\sem[M_2]{\Fr(t)}$ for all \spp{} expressions, and
\item $\sem[M_1]{I} \subseteq \sem[M_2]{I}$ for all
inductive relations $I \in \I$.
\end{itemize}
The relation $\lei$ is a partial order.

We say that $M \in \Mod(\hat{M})$ is a \emph{frame model} if its
interpretation function $\sem\cdot$ satisfies the inductive equations and the \spp{} equations, and furthermore
\begin{enumerate}
\item each $M' \in \Mod(\hat{M})$ with $M' \slef M$ does not satisfy the \spp{} equations, and
\item each $M' \in \Mod(\hat{M})$ with $M' \slei M$ does not satisfy the inductive equations.
\end{enumerate}



For proving the existence of a unique frame model, we use the following lemma for dealing with guards and terms with mutable functions.
\begin{lemma}\label{lem:simple-formulas}
Let $\hat{M}$ be a pre-model, $M_1,M_2 \in \Mod(\hat{M})$, and $\nu$ be a variable assignment.
\begin{enumerate}
\item If $\varphi$ is formula that does not use inductive relations and \spp{} expressions, then $M_1, \nu \models \varphi$ iff  $M_2,\nu \models \varphi$.
\item If $t$ is a term that has no \spp{} expressions as subterms, then $\sem[M_1,\nu]{t} = \sem[M_2,\nu]{t}$.
\item If  $t = f(t_1, \ldots,t_n)$ is a term with a mutable function symbol $f \in \Fm$,
then $\sem[M_1,\nu]{t_i} = \sem[M_2,\nu]{t_i}$ for all $i$.
\end{enumerate}
\end{lemma}
\begin{proof}
Parts 1 and 2 are immediate from the fact that $M_1$ and $M_2$ only differ in the interpretation of the inductive relations and \spp{} expressions. For the third claim, note that we assumed that the only functions involving arguments of sort $\sfs$ are the standard functions for set manipulation. Hence, a term build from a mutable function symbol cannot have \spp{} expressions as subterms. Therefore, the third claim follows from the second one.
\end{proof}

\revII{We now prove Proposition~\ref{pro:unique-frame-model} (see Section~\ref{sec:frames}). We first restate the proposition formally using the terminology developed above.}
\begin{proposition}\label{pro:unique-lfp}
For each pre-model $\hat{M}$, there is a unique frame model in
$\Mod(\hat{M})$.
\end{proposition}
\begin{proof}
The \spp{} equations define an operator $\muf$ on $\Mod(\hat{M})$. This operator $\muf$ is defined in a
standard way, as explained in the following. Let $M \in
\Mod(\hat{M})$. Then $\muf(M)$ is a model in $\Mod(\hat{M})$ where $\sem[\muf(M)]{\Fr(\varphi)}$, resp.\ $\sem[\muf(M)]{\Fr(t)}$, is
obtained by taking the right-hand side of the corresponding
equation. For example, $\sem[\muf(M)]{\Fr(\varphi_1 \land \varphi_2)}(\nu) =
\sem{\Fr(\varphi_1)}(\nu) \cup \sem{\Fr(\varphi_2)}(\nu)$.
The interpretation of the inductive predicates is left unchanged by $\muf$.

We can show that $\muf$ is a monotonic operator on $(\Mod(\hat{M}), \lef)$, that is,  for all $M_1,M_2 \in \Mod(\hat{M})$ with $M_1 \lef M_2$ we have that $\muf(M_1) \le \muf(M_2)$.
It is routine to check monotonicity of $\muf$ by induction on the structure of the \spp{} expressions. We use Lemma~\ref{lem:simple-formulas} for the only cases in which the semantics of formulas and terms is used in the \spp{} equations, namely $\ite$-formulas, existential formulas, and terms $f(t_1, \ldots, t_n)$ with mutable function $f$.
Consider, for example, the \spp{} equation 
\[
\begin{array}{l}
\sem{\Fr(f(t_1, \ldots, t_n))}(\nu)\\ {} =
  \bigcup\limits_{i \text{ with } t_i \text{ of sort } \fs}\{\semnu{t_i}\} \cup \bigcup\limits_{i=1}^n \sem{\Fr(t_i)}(\nu)
\end{array}
\]
for $f \in \Fm$, and let $M_1 \lef M_2$ be in $\Mod(\hat{M})$ and
$\nu$ be a variable assignment. Then
\[
\begin{array}{l}
 \sem[\muf(M_1)]{\Fr(f(t_1, \ldots, t_n))}(\nu) \\
  = \bigcup\limits_{i \text{ with } t_i \text{ of sort } \fs}
  \{\sem[M_1,\nu]{t_i}\} \cup \bigcup\limits_{i=1}^n \sem[M_1]{\Fr(t_i)}(\nu)\\
  \stackrel{(1)}{=}
  \bigcup\limits_{i \text{ with } t_i \text{ of sort } \fs}
  \{\sem[M_2,\nu]{t_i}\} \cup \bigcup\limits_{i=1}^n \sem[M_1]{\Fr(t_i)}(\nu)\\
  \stackrel{(2)}{\subseteq}
  \bigcup\limits_{i \text{ with } t_i \text{ of sort } \fs}
  \{\sem[M_2,\nu]{t_i}\} \cup \bigcup\limits_{i=1}^n \sem[M_2]{\Fr(t_i)}(\nu)\\
  =  \sem[\muf(M_2)]{\Fr(f(t_1, \ldots, t_n))}(\nu)
  \end{array}
\]
where (1) holds because of Lemma~\ref{lem:simple-formulas}, and (2) holds because $M_1 \lef M_2$.

As a further case, consider the \spp{} equation for $R(\overline{t})$ where $R$ is an inductively defined relation and $t = (t_1, \ldots, t_n)$.
\[
\begin{array}{l}
 \sem[\muf(M_1)]{\Fr(R(\overline{t}))}(\nu) \\
  = \sem[M_1]{\Fr(\rho_R(\overline{x})}(\nu[\overline{x} \leftarrow \sem[M_1,\nu]{\overline{t}}]) \cup \bigcup\limits_{i=1}^n \sem[M_1]{\Fr(t_i)}(\nu)\\
  \stackrel{(*)}{\subseteq} \sem[M_2]{\Fr(\rho_R(\overline{x})}(\nu[\overline{x} \leftarrow \sem[M_2,\nu]{\overline{t}}]) \cup \bigcup\limits_{i=1}^n \sem[M_2]{\Fr(t_i)}(\nu)\\
  = \sem[\muf(M_2)]{\Fr(R(\overline{t}))}(\nu)
   \end{array}
\]
For the inclusion $(*)$ we use the fact that the $t_i$ do not contain \spp{} expressions as subterms by our restriction of the type of inductively defined relations. Hence, by Lemma~\ref{lem:simple-formulas}, $\sem[M_1,\nu]{t_i} = \sem[M_2,\nu]{t_i}$. 

Similarly, one can show the inclusion for the other \spp{} equations.

We also obtain an operator $\mui$ from the inductive equations, 
which leaves the interpretation of the \spp{} expressions
unchanged. The operator $\mui$ is monotonic on $(\Mod(\hat{M}),\lei)$
because inductive predicates can only be used positively in the
inductive definitions, and furthermore $\lei$ only compares models
with the same interpretation of the \spp{} expressions.

In order to obtain the unique frame model, we first consider the
subset of $\Mod(\hat{M})$ in which all inductive predicates are
interpreted as empty set. On this set of models, $\lef$ is a partial order 
and forms
a complete lattice (the join and meet for the lattice are
obtained by taking the pointwise union, respectively intersection, of
the interpretations of the \spp{} expression). By the Knaster-Tarski
theorem, there is a unique least fixpoint of $\muf$. This fixpoint can be obtained
by iterating $\muf$ starting from the model in $\Mod(\hat{M})$ that
interprets all inductive relations and the \spp{} expression by the
empty set (in general, this iteration is over the ordinal numbers, not just the natural numbers). 
Let $M_f$ be this least fixpoint.

The subset of $\Mod(\hat{M})$ in which the \spp{} expressions are interpreted as in $M_f$ forms
a complete lattice with the partial order $\lei$. Again by the Knaster-Tarski Theorem, there is a unique least fixpoint.
This least fixpoint can be obtained by iterating the operator $\mui$ starting from $M_f$ (again, the iteration is over the ordinals).

Denote the resulting model by $M_{f,i}$. It interprets the \spp{} expressions in the same way as $M_f$, and thus $M_f \lef M_{f,i}$ and $M_{f,i} \lef M_i$. By monotonicity of $\muf$, $M_{f,i}$ is also a fixpoint of $\muf$ and thus satisfies the \spp{} equations. Hence $M_{f,i}$ satisfies the inductive equations and the \spp{} equations. It can easily be checked that $M_{f,i}$ also satisfies the other conditions of a frame model:
Let $M \in \Mod(\hat{M})$ with $M \slef M_{f,i}$. Then also $M \lef M_f$ and assuming that $M$ satisfies the \spp{} equations yields a smaller fixpoint of $\muf$, and thus a contradiction. Similarly, a model $M \slei M_{f,i}$ cannot satisfy the inductive equations.

It follows that 
$M_{f,i}$ is a frame model in $\Mod(\hat{M})$. Uniqueness follows from the uniqueness of the least fixpoints of $\muf$ and $\mui$ as used in the construction of $M_{f,i}$.
\end{proof}

\subsection*{\revII{Consequences of Frame Logic Semantics}}
\label{sec:consequences}

\revII{We now highlight some consequences of the semantics of FL and highlight differences with other logics.}

\revII{\textbf{Global Heap Semantics vs Local Heap Semantics:}}
\revII{Our semantics is \emph{global} in the sense that formulas are interpreted on the entire model. Supports of formulas define a subset of the model that capture the part on which their depends on, and roughly correspond to \emph{heaplets} in separation logic. Separation logic semantics is however defined not on the global heap but with respect to local heaplets (i.e., captured using rules of the form
$s, h \models \varphi$), where $h$ is a heaplet~\cite{reynolds02}, and hence corresponding to a local heap semantics.}

\revII{Consider the FL formula $\exists y: y=f(x).\, \top$. This formula is valid in FL but the corresponding SL formula $\exists y.\, x \overset{f}{\mapsto} y$ is not valid; it holds on the heaplet $h$ where $h$ is a singleton location containing the valuation of $x$, but not on other heaplets. 
Note that $Sp(\exists y: y=f(x).\, \top)$
will be this heaplet $h$.}




\revII{\textbf{Substitutions:}} \revII{Another interesting consequence of FL semantics is that substituting a formula with another equivalent formula (in terms of truthhood) may not result in equivalent formulas.}


\revII{Consider the formulas $\top$ and $f(x) = f(x)$ for a mutable function $f$. These are equivalent in FL in terms of truthhood (they are both valid, holding in all models). 
Consider the valid formula $\Fr(f(x) = f(x)) = \{x\}$. 
Substituting $f(x)=f(x)$ with $\top$ will result in an invalid formula $\Fr(\top)=\{x\}$.}

\revII{More precisely, there can be formulas that are equivalent in terms of truthhood (in the global heap) but have different supports, and hence one may not be substituted by another in every context. However, it is true that if $\alpha$ and $\beta$ are equivalent in terms of truthhood \emph{and} have the same supports in every model, then they can be substituted for each other. This motivates a stronger definition of equivalence for FL formulas:}



\begin{definition}[\revII{Equivalence of FL formulas}]
\label{def:fl-equiv}
\revII{FL formula $\varphi_1$ and $\varphi_2$ are said to be equivalent, denoted $\varphi_1 \equiv \varphi_2$ if for every frame model $M$ and interpretation of free variables $\nu$, $M, \nu \vDash \varphi_1$ iff $M, \nu \vDash \varphi_2$ \emph{and} $\semnu{\Fr(\varphi_1)} = \semnu{\Fr(\varphi_2)}$.}
\end{definition}

\revII{Formulas are not equivalent if they merely imply each other; it must also be the case that their supports are equal. 
Reasoning directly in FL requires care in order to cater the above notion. However, as we point out in Section~\ref{sec:reasoning-with-fl}, we can convert FL to FO-RD and then use standard reasoning for FO-RD (including normal substitution that FO-RD allows).}




\revII{\textbf{Supports and Heaplets:}}
\revII{In FL, our design decision is to have precisely \emph{one} support for any formula. In separation logic, there can several heaplets under which a formula can hold true. For example, the formula $\top$ holds in \emph{any} heaplet in separation logic, while we have chosen the support of $\top$ (and that of $\bot$) to be $\emptyset$.}

\revII{Note that $\alpha \implies (\alpha \wedge \top)$ and 
$\alpha \implies (\alpha \vee \bot)$ are both valid in FL.
And further, $\Fr(\alpha) = \Fr(\alpha \wedge \top) = \Fr(\alpha \vee \bot)$, by our semantics.
} 

\revII{Notice that $\alpha \wedge \neg(\alpha) \implies \bot$ is valid in FL, despite the fact that the support of the antecedent and that of the consequent may be different (unlike separation logic). Similarly, $\alpha \vee \neg \alpha \implies \top$ is also valid.}

\subsection{A Frame Theorem}\label{sec:frame-theorem}
The \spp{} of a formula can be used for frame reasoning in the following sense: if we modify
a model $M$ by changing the interpretation of the mutable functions
(e.g., a program modifying pointers), then truth values of formulas do not change if the change happens outside the \spp{} of
the formula. This is formalized and proven below.

Given two models $M,M'$ over the same universe, we say that $M'$ is a
\emph{mutation of $M$} if $\sem{R} = \sem[M']{R}$, $\sem{c} = \sem[M']{c}$, and $\sem{f} = \sem[M']{f}$,
for all constants $c$, relations $R \in \R$, and
functions $f \in F \setminus \Fm$. In other words, 
$M$ can only be different from $M'$ on the interpretations of the mutable functions, the inductive relations, and the \spp{} expressions.

Given a subset $X \subseteq U_{\fs}$ of the elements from the
foreground universe, we say that the \emph{mutation is stable on $X$}
if the values of the mutable functions did not change on arguments
from $X$, that is, $\sem{f}(u_1, \ldots, u_n) =
\sem[M']{f}(u_1, \ldots, u_n)$ for all mutable functions $f \in \Fm$ and
all appropriate tuples $u_1, \ldots, u_n$ of arguments with $\{u_1, \ldots,
u_n\} \cap X \not= \emptyset$.

\begin{theorem}[Frame Theorem]\label{thm:frametheorem}
Let $M,M'$ be frame models such that $M'$ is a mutation of $M$ that is
stable on $X \subseteq U_{\fs}$, and let $\nu$ be a variable assignment. 
Then $M,\nu \models \alpha$ iff $M',\nu \models \alpha$ for all formulas $\alpha$
with $\sem{\Fr(\alpha)}(\nu) \subseteq X$, and $\semnu{t} = \sem[M',\nu]{t}$ for all terms $t$ with 
$\sem{\Fr(t)}(\nu) \subseteq X$.
\end{theorem}

We dedicate the rest of this subsection to the proof of the Frame Theorem. 
\begin{proof}
  The intuition behind the statement of the theorem should be
  clear. The \spp{} of a formula/term contains the elements on which
  mutable functions are dereferenced in order to evaluate the formula/term. 
  If the mutable functions do not change on this set,
  then the evaluation does not change.

  %
  For a formal proof of the Frame Theorem, we refer to the terminology and definitions introduced in
  Section~\ref{sec:frame-models}, and to the proof of
  Proposition~\ref{pro:unique-lfp} in Section~\ref{sec:frame-models}, in which the unique frame model is obtained 
  by iterating the operators $\muf$ and $\mui$, which are defined by the \spp{} equations 
  and the inductive equations. 
  
  In general, this iteration of the operators ranges over ordinals (not just natural numbers).
  For an ordinal $\eta$, let $M_\eta$ and $M'_\eta$ be the models at step $\eta$
  of the fixpoint iteration for obtaining the frame models $M$ and
  $M'$. So the sequence of the $M_\eta$ have monotonically increasing 
  interpretations of the inductive relations and \spp{} expressions, and
  are equal to $M$ on the interpretation of the other relations and functions.
  The frame model $M$ is obtained at some stage $\xi$ of the fixpoint 
  iteration, so $M = M_\xi$. More precisely, the frame model is contructed by first
  iterating the operator $\muf$ until the fixpoint of the \spp{} expressions is reached.
  During this iteration, the inductive relations are interpreted as empty. 
  Then the operator $\mui$ is iterated until also the inductive relations reach their fixpoint.
  Below, we do an induction on $\eta$. In that induction, we do not explicitly distinguish these two phases
  because it does not play any role for the arguments (only in one place and we mention it explicitly there).
  
  By induction on $\eta$, we can show that $M_\eta,\nu \models
  \varphi \Leftrightarrow 
  M'_\eta,\nu \models \varphi$, and $\sem[M_\eta,\nu]{t} =
  \sem[M'_\eta,\nu]{t}$ for all variable assignments $\nu$ and all 
  formulas $\varphi$  with $\sem{\Fr(\varphi)}(\nu) \subseteq
  X$, respectively terms $t$ with $\sem{\Fr(t)}(\nu) \subseteq X$. For each $\eta$, we
  furthermore do an induction on the structure of the formulas,
  respectively terms.

  Note that the assumption that the \spp{} is contained in $X$ refers
  to the \spp{} in $M$. So when applying the induction, we have to
  verify that the condition on the \spp{} of a formula/term is
  satisfied in $M$ (and not in $M_\eta$).

  For the formulas, the induction is straight forward, using
  Lemma~\ref{lem:simple-formulas} in the cases of existential
  formulas and $ite$-formulas.
  Consider, for example, the case of an existential formula $\psi =
  \exists y:\gamma.\varphi$ with $\sem{\Fr(\psi)} \subseteq X$.
  \[
  \begin{array}{ll}
  & M_\eta,\nu \models \exists y:\gamma.\varphi \\ 
  \Leftrightarrow&
  \mbox{exists } u \in D_y: M_\eta,\nu[y \leftarrow u] \models \gamma\\
  & \mbox{ and } M_\eta,\nu[y \leftarrow u] \models \varphi \\
   \stackrel{(*)}{\Leftrightarrow}& 
  \mbox{exists } u \in D_y: M_\eta',\nu[y \leftarrow u] \models \gamma \\
  & \mbox{ and } M_\eta',\nu[y \leftarrow u] \models \varphi \\
   \Leftrightarrow & M_\eta',\nu \models \exists y:\gamma.\varphi 
  \end{array}
  \]
  where $(*)$ holds by induction on the structure of the formula. We
  only have to verify that $\sem{\Fr(\gamma)}(\nu[y \leftarrow u])
  \subseteq X$ and $\sem{\Fr(\varphi)}(\nu[y \leftarrow u])
  \subseteq X$ in order to use the induction hypothesis.

  Since $\gamma$ is a guard of an existential formula, it satisfies
  the condition of Lemma~\ref{lem:simple-formulas}, and therefore its
  truth value is the same in $(M_\eta,\nu[y \leftarrow u])$ for all
  ordinals $\eta$ (Lemma~\ref{lem:simple-formulas} applies because all the models $M_\eta$ 
  differ only in the interpretations of the
  \spp{} expressions and inductive relations, and thus have the same pre-model). 
  In particular, $M,\nu[y \leftarrow u] \models
  \gamma$ since $M = M_\xi$ for some ordinal $\xi$. From the equations
  for the \spps{} we obtain $\sem{\Fr(\gamma)}(\nu[y \leftarrow u])
  \subseteq \sem{\Fr(\psi)}(\nu)$ and $\sem{\Fr(\varphi)}(\nu[y \leftarrow
    u]) \subseteq \sem{\Fr(\psi)}(\nu)$. The desired claim now follows
  from the fact that $\sem{\Fr(\psi)}(\nu) \subseteq X$.

  For inductive relations $R$ with definition $R(\overline{x}) := \rho_R(\overline{x})$, we have to use the induction on the ordinal
  $\eta$. Assume that $\varphi = R(\overline{t})$ for $\overline{t} = (t_1, \ldots, t_n)$, and that
  $\sem{\Fr(\varphi)}(\nu) \subseteq X$. Then $\sem{\Fr(\rho_R(\overline{x})}(\nu[\overline{x} \leftarrow \overline{t}])
  \subseteq X$ and  $\sem{\Fr(t_i)}(\nu) \subseteq X$ for all $i$ by the
  \spp{} equations.

  For the case of a limit ordinal $\eta$, the inductive relations of
  $M_\eta$, resp.~$M'_\eta$, are
  obtained by taking the union of the interpretations of the
  inductive relations for all $M_\zeta$, resp.~$M'_\zeta$, for all $\zeta <
  \eta$. So the claim follows directly by induction.

  For a successor ordnial $\eta+1$, we can assume that we are in the second phase of the construction
  of the frame model (the iteration of the operator $\mui$). For the first phase the claim trivially holds because
  all the inductive relations are interpreted as empty. Thus, we have
   \[
  \begin{array}{rcl}
  M_{\eta+1},\nu \models R(\overline{t}) &\Leftrightarrow &
    M_\eta,\nu \models \rho_R(\overline{t}) \\
    &\Leftrightarrow &
    M_\eta,\nu[\overline{x} \leftarrow \sem[M_\eta,\nu]{\overline{t}}] \models \rho_R(\overline{x}) \\
  &\stackrel{(*)}{\Leftrightarrow} &
  M_\eta',\nu[\overline{x} \leftarrow \sem[M_\eta',\nu]{\overline{t}}] \models \rho_R(\overline{x}) \\
   &\Leftrightarrow & M_{\eta+1}',\nu \models R(\overline{t})
  \end{array}
  \]
  where $(*)$ holds by induction on $\eta$.
  
  The other cases for formulas are similar (or simpler).

  Concerning the terms, we also present some cases only, the other
  cases being similar or simpler.

  We start with the case $t = f(t_1, \ldots, t_n)$ for a mutable
  function $f$. Let $\nu$ be a variable assignment with $\semnu{t}
  \subseteq X$. By the \spp{} equations, $\semnu{t_i}
  \subseteq X$ for all $i$. We have $\sem[M_\eta,\nu]{t} =
  \sem{f}(\sem[M_\eta,\nu]{t_1}, \ldots, \sem[M_\eta,\nu]{t_n})$. By
  induction on the structure of terms, we have $\sem[M_\eta,\nu]{t_i}
  = \sem[M_\eta',\nu]{t_i} =: u_i$. 
  By Lemma~\ref{lem:simple-formulas}, we conclude that
  $\sem[M_\eta,\nu]{t_i} = \semnu{t_i}$. Since $f$ is mutable, it contains at least one argument of sort $\fs$, say $t_j$. Then $\semnu{t_j} \in
  \sem{\Fr(t)}(\nu) \subseteq X$, and the mutation did not change the function value of $f$ on
  the tuple $(u_1, \ldots, u_n)$. So 
  we obtain in summary that $\sem[M',\nu]{t}
  = \sem[M']{f}(u_1,
  \ldots, u_n) = \sem{f}(u_1,\ldots,u_n) =\semnu{t}$.

  Now consider terms of the form $\Fr(\varphi)$. We need to proceed by
  induction on the structure of $\varphi$. We present the case of
  $\varphi = \ite(\gamma:\varphi_1,\varphi_2)$.  Let $\nu$ be a
  variable assignment with $\sem{\Fr(\varphi)}(\nu) \subseteq X$. Assume
  that $M_\eta,\nu \models \gamma$. By the condition on guards,
  Lemma~\ref{lem:simple-formulas} yields that $M,\nu \models \gamma$
  and thus $\sem{\Fr(\gamma)}(\nu) \subseteq X$ and
  $\sem{\Fr(\varphi_1)}(\nu) \subseteq X$. We obtain
  \[
  \begin{array}{rcl}
    \sem[M_\eta]{\Fr(\varphi)}(\nu) &=& \sem[M_\eta]{\Fr(\gamma)}(\nu)
    \cup  \sem[M_\eta]{\Fr(\varphi_1)}(\nu) \\
    &\stackrel{(*)}{=}& \sem[M_\eta']{\Fr(\gamma)}(\nu)
    \cup  \sem[M_\eta']{\Fr(\varphi_1)}(\nu) \\
    &=& \sem[M_\eta']{\Fr(\varphi)}(\nu)
  \end{array}
  \]
  where $(*)$ follows by induction on the structure of the formula inside the \spp{} expression. The case $M_\eta,\nu \not\models
  \gamma$ is analogous.

  Now consider $\Fr(\varphi)$ with $\varphi = R(\overline{t})$ for an inductively defined relation $R$ with definition $R(\overline{x}) = \rho_R(\overline{x})$ and $\overline{t} = (t_1, \ldots, t_n)$.  Let $\nu$ be a
  variable assignment with $\sem{\Fr(\varphi)}(\nu) \subseteq X$. 
  By the \spp{} equations, $\sem{\Fr(\rho_R(\overline{x}))}(\nu[\overline{x} \leftarrow \semnu{\overline{t}}]) \subseteq X$ and $\sem{\Fr(t_i)}(\nu) \subseteq X$.
  
  Let $\eta+1$ be a successor ordinal. 
 Then
 \[
\begin{array}{l}
 \sem[M_{\eta+1}]{\Fr(R(\overline{t}))}(\nu) \\
  = \sem[M_{\eta}]{\Fr(\rho_R(\overline{x}))}(\nu[\overline{x} \leftarrow \sem[M_\eta,\nu]{\overline{t}}]) \cup \bigcup\limits_{i=1}^n \sem[M_\eta]{\Fr(t_i)}(\nu)\\
  \stackrel{(*)}{=} 
    \sem[M_{\eta}']{\Fr(\rho_R(\overline{x}))}(\nu[\overline{x} \leftarrow \sem[M_\eta',\nu]{\overline{t}}]) \cup \bigcup\limits_{i=1}^n \sem[M_\eta']{\Fr(t_i)}(\nu)\\
 = \sem[M_{\eta+1}']{\Fr(R(\overline{t}))}(\nu)
   \end{array}
\]
where $(*)$ holds by induction on $\eta$. We can apply the induction hypothesis because the terms $t_i$ do not contain
\spp{} expressions by the restriction on the type of inductive relations, and thus $\sem[M_\eta',\nu]{\overline{t}} = \sem[M_\eta,\nu]{\overline{t}} = \semnu{\overline{t}}$ by Lemma~\ref{lem:simple-formulas}.

  The proof of the other cases works in a similar fashion.
  \end{proof}
  





\subsection{Reduction from Frame Logic to FO-RD} \label{sec:reduction-to-FO-RD}

The only extension of frame logic compared to FO-RD is the operator $\Fr$, which defines a function from interpretations of free variables to sets of foreground elements. The semantics of this operator can be captured within FO-RD itself, so reasoning within frame logic can be reduced to reasoning within FO-RD. 

A formula $\alpha(\overline{y})$ with $\overline{y} = y_1, \ldots,y_m$ has one \spp{} for each interpretation of the free variables.
We capture these \spps{} by an inductively defined
 relation $\Fr_\alpha(\overline{y},z)$ of arity $m+1$ such that for each
frame model $M$, we have $(u_1, \ldots,u_m,u) \in \sem{\Fr_\alpha}$
if $u \in \sem{\Fr(\alpha)}(\nu)$ for the interpretation $\nu$ that
interprets $y_i$ as $u_i$.

Since the semantics of $\Fr(\alpha)$ is defined over the structure of
$\alpha$, we introduce corresponding inductively defined relations $\Fr_\beta$ and $\Fr_t$ for all
subformulas $\beta$ and subterms $t$ of either $\alpha$ or of a formula $\rho_R$ for $R \in \mathcal{I}$.


\begin{figure*}
\[
  \begin{array}{rcl}
    \Fr_c(\overline{y},z) &:=& \bot ~~~~\mbox{ for a constant $c$}\\[3pt]
    \Fr_x(\overline{y},z) &:=& \bot ~~~~\mbox{ for a variable $x$}\\[3pt]
    \Fr_{f(t_1, \ldots, t_n)}(\overline{y},z) &:=& 
  \begin{cases}
  \left(\bigvee\limits_{t_i \text{ of sort } \fs} z = t_i\right) \lor \bigvee\limits_{i=1}^n \Fr_{t_i}(\overline{y},z) & \mbox{if } f \in \Fm \\[3pt]
  \bigvee\limits_{i=1}^n \Fr_{t_i}(\overline{y},z) & \mbox{if } f \not\in \Fm 
  \end{cases}\\[3pt]
  \Fr_{\Fr(\beta)}(\overline{y},z) &:=& \Fr_{\beta}(\overline{y},z)\\[3pt]
  \Fr_{R(t_1, \ldots, t_n)}(\overline{y},z) &:=& \bigvee_{i=1}^n \Fr_{t_i}(\overline{y},z)  \mbox{
  for $R \in \R$}\\[3pt]
 \revII{\Fr_{(\top)}(\overline{y},z)} &\revII{:=}& \revII{\bot} \\[3pt]
  \Fr_{(\bot)}(\overline{y},z) &:=& \bot \\[3pt]
 \Fr_{(t_1 = t_2)}(\overline{y},z) &:=& \Fr_{t_1}(\overline{y},z) \lor \Fr_{t_2}(\overline{y},z) \\[3pt]
 \Fr_{R(\overline{t})}(\overline{y},z)  &:=& \Fr_{\rho_{R(\overline{x})}}(\overline{y},z)[\overline{t}/\overline{x}] \lor \bigvee_{i=1}^n \Fr_{t_i}(\overline{y},z)\\[3pt]
 && \mbox{for $R \in \I$ with definition $R(\overline{x}) := \rho_R(\overline{x})$} \\[3pt]
\revII{\Fr_{\beta_1 \land \beta_2}(\overline{y},z)  }&\revII{:=}& \revII{\Fr_{\beta_1}(\overline{y},z) \lor \Fr_{\beta_2}(\overline{y},z)}\\[3pt]
\Fr_{\lnot \beta}(\overline{y},z) &:=& \Fr_{\beta}(\overline{y},z)\\[3pt]
\Fr_{\ite(\gamma:\beta_1,\beta_2)}(\overline{y},z)  &:=&
\Fr_{\gamma}(\overline{y},z) \lor
\ite(\gamma:\Fr_{\beta_1}(\overline{y},z),\Fr_{\beta_2}(\overline{y},z)) \\[3pt]
\Fr_{\ite(\gamma:t_1,t_2)}(\overline{y},z)  &:=& \Fr_{\gamma}(\overline{y},z)
\lor \ite(\gamma:\Fr_{t_1}(\overline{y},z),\Fr_{t_2}(\overline{y},z)) \\[3pt]
\Fr_{\exists x: \gamma.\beta}(\overline{y},z) &:=& \exists
x. (\Fr_{\gamma}(\overline{y},z) \lor (\gamma \land \Fr_{\beta}(\overline{y},z)))
  \end{array}
\]
\caption{Translation of \spp{} equations to FO-RD.}
\label{fig:translation-to-ford}
\end{figure*}

The equations for \spps{} from Figure~\ref{fig:frame-equations} can be expressed by inductive definitions for the relations $\Fr_\beta$. The translations are shown in Figure~\ref{fig:translation-to-ford}. 
For the definitions, we assume that $\overline{y}$ contains all
variables that are used in $\alpha$ and in the formulas $\rho_R$ of the inductive definitions. We further assume that 
each variable either is used in at most one of the formulas $\alpha$ or $\rho_R$, and either only occurs freely in it, or is quantified
at most once. The relations
$\Fr_\beta$ are all of arity $m+1$, even if the subformulas do not use
some of the variables. In practice, one would rather use relations of
arities as small as possible, referring only to the relevant
variables. In a general definition, this is, however, rather
cumbersome to write, so we use this simpler version in which we do not
have to rearrange and adapt the variables according to their use in
subformulas.

For the definition of $\Fr_{R(\overline{t})}$ where $R$ is an inductively defined relation, note that the variables $\overline{x}$ from the definition of $R$ are contained in $\overline{y}$ by the above assumptions, and are substituted by the terms in $\overline{t}$ in the first part of the formula.
Similiarly, the quantified variable $x$ in an existential formula is contained in $\overline{y}$.

\begin{figure}
\begin{align*}
    \lst(x) :=\; &\ite(x = nil, \top, \exists z : z = \nxt(x).\; \lst(z) \land x \not\in \Fr(\lst(z)) \tag{linked list} \\
    \dll(x) :=\; &\ite(x = nil : \top, \ite(\nxt(x) = nil : \top, \exists z : z = \nxt(x). \\
       &\prv(z) = x \land \dll(z) \land x \not\in \Fr(\dll(z)))) \tag{doubly linked list} \\
    \lseg(x,y) :=\; &\ite(x = y : \top, \exists z : z = \nxt(x).\; \lseg(z,y) \land x \not\in \Fr(\lseg(z, y))) \tag{linked list segment} \\
    \len(x,n) := \; &\ite(x = nil : n = 0, \exists z : z = \nxt(x).\; \len(z, n - 1)) \tag{length of list} \\
    \slst(x) :=\; &\ite(x = nil : \top, \ite(\nxt(x) = nil, \top, \exists z : z = \nxt(x). \\
    &\key(x) \leq \key(z) \land \slst(z) \land x \not\in \Fr(\slst(z)))) \tag{sorted list} \\
    \mkee(x, M) :=\; &\ite(x = nil : M = \emptyset, \exists z, M_1 : z = \nxt(x). \\
    &M = M_1 \cup_{m} \{\key(x)\} \land \mkee(z, M_1)) \land x \not\in \Fr(\mkee(z, M_1)) \tag{multiset of keys in linked list} \\
    \btree(x) :=\; &\ite(x = nil : \top, \exists \ell, r : \ell = \lft(x) \land r = \rght(x). \\
    &\btree(\ell) \land \btree(r) \land x \not\in \Fr(\btree(\ell)) \land x \not\in \Fr(\btree(r))\; \land \\
    &\Fr(\btree(\ell)) \cap \Fr(\btree(r)) = \emptyset) \tag{binary tree} \\
    \revII{\minkey(x, n) :=\; }&\revII{\ite(x = nil: \infty, \exists \ell, r, n_1, n_2 : \ell = \lft(x) \land r = \rght(x)\, \land} \\
    & \revII{n_1 = \minkey(l) \land n_2 = \minkey(r).\, n = \mathit{min}(n_1,\key(x),n_2))} \tag{\revII{minimum key in a tree/dag}}\\
    \revII{\maxkey(x, n) :=\; }&\revII{\ite(x = nil: -\infty, \exists \ell, r, n_1, n_2 : \ell = \lft(x) \land r = \rght(x)\, \land} \\
    & \revII{n_1 = \maxkey(l) \land n_2 = \maxkey(r).\, n = \mathit{max}(n_1,\key(x),n_2))} \tag{\revII{maximum key in a tree/dag}}\\
    \bst(x) :=\; &\ite(x = nil : \top, \exists \ell, r,\mathit{maxl},\mathit{minr} : \ell = \lft(x) \land r = \rght(x)\, \land \\ 
    & \maxkey(l,\mathit{maxl}) \land \minkey(r,\mathit{minr}).\,\\
    & \mathit{maxl} \leq \key(x) \land \key(x) \leq \mathit{minr}\; \land \\ 
    &\bst(\ell) \land \bst(r) \land x \not\in \Fr(\bst(\ell)) \land x \not\in \Fr(\bst(r))\, \land \\
    &\Fr(\bst(\ell)) \cap \Fr(\bst(r)) = \emptyset)))) \tag{binary search tree} \\
    \height(x, n) :=\; &\ite(x = nil : n = 0, \exists \ell, r, n_1, n_2: \ell = \lft(x) \land r = \rght(x). \\
    &\height(\ell, n_1) \land \height(r, n_2) \land \ite(n_1 > n_2 : n = n_1 + 1, n = n_2 + 1)) \tag{height of binary tree} \\
    \bfac(x, b) :=\; &\ite(x = nil : 0, \exists \ell, r, n_1, n_2 : \ell = \lft(x) \land r = \rght(x). \\ 
    &\height(\ell, n_1) \land \height(r, n_2) \land b = n_2 - n_1) \tag{balance factor (for AVL tree)} \\
    \avl(x) :=\; &\ite(x = nil : \top, \exists \ell, r : \ell = \lft(x) \land r = \rght(x). \\
    &\avl(\ell) \land \avl(r) \land \bfac(x) \in \{-1,0,1\}\; \land \\
    & x \not\in \Fr(\avl(\ell)) \cup \Fr(\avl(r)) \land \Fr(\avl(\ell)) \cap \Fr(\avl(r)) = \emptyset) \tag{avl tree} \\
    \ttree(x) :=\; &\pttree(x, nil) \tag{threaded tree} \\
    \pttree(x,p) :=\;& \ite(x = nil : \top, \exists \ell, r : \ell = \lft(x) \land r = \rght(x). \\
    &((r = nil \land \tnext(x) = p) \lor (r \neq nil \land \tnext(x) = r))\; \land \\
    &\pttree(\ell, x) \land \pttree(r, p) \land x \not\in \Fr(\pttree(\ell, x)) \cup \Fr(\pttree(r, p))\; \land \\
    &\Fr(\pttree(\ell, x)) \cap \Fr(\pttree(r, p)) = \emptyset) \tag{threaded tree auxiliary definition}
\end{align*}
\caption{Example definitions of data-structures and other predicates in Frame Logic}
\label{fig:rec-defs-list}
\end{figure}

It is not hard to see that general frame logic formulas can be translated to FO-RD formulas that make use of these new inductively defined relations.

\begin{proposition} 
For every frame logic formula there is an equisatisfiable FO-RD formula
with the signature extended by auxiliary predicates for recursive definitions of \spps{}. 
\end{proposition}


%

\subsection{Expressing Data-Structures Properties in FL}

We now present the formulation of several data-structures and properties about them in FL. 
Figure~\ref{fig:rec-defs-list} depicts formulations of singly- and doubly-linked lists, list segments, lengths of lists, sorted lists, the multiset of keys stored in a list (assuming a background sort of multisets), binary trees, their heights, and AVL trees. In all these definitions, the \spp{} operator plays a crucial role. We also present a formulation of \emph{single threaded binary trees} (adapted from~\cite{brinck81}), which are binary trees where, apart from tree-edges, there is a pointer $tnext$ that connects every tree node to the inorder successor in the tree; these pointers go from leaves to ancestors arbitrarily far away 
in the tree, making it a nontrivial definition. 

We believe that FL formulas naturally and succinctly express these data-structures and their properties, making it an attractive logic for annotating programs.

\section{Programs and Proofs}\label{sec:prog-and-prfs}

In this section, we develop a program logic for a while-programming language that can destructively update heaps.
%
We assume that location variables are denoted by variables of the form $x$ and $y$, whereas variables that denote other data (which would correspond to the \emph{background} sorts in our logic) are denoted by $v$. We omit the grammar to construct background terms and formulas, and simply denote such `background expressions' with $be$ and clarify the sort when it is needed. Finally, we assume that our programs are written in Single Static Assignment (SSA) form, which means that every variable is assigned to at most once in the program text. The grammar for our programming language is in Figure~\ref{fig:prog-lang}.

\begin{figure}
\begin{align*}
    S &::= x\, :=\, c\;\; |\;\; x\, :=\, y\;\; |\;\; x\, :=\, y.f\;\; |\;\; v\, :=\, be\;\; |\;\; x.f\, :=\, y
    \\ &\;\;|\; \;\mathsf{alloc}(x)\;\; |\;\; \mathsf{free}(x) 
    \;\; |\;\; \mathsf{if}\; be\; \mathsf{then}\; S\; \mathsf{else}\; S\;\;
    |\;\; \mathsf{while}\; be\; \mathsf{do}\; S\; 
    |\;\; S\, ; \, S
\end{align*}
\caption{Grammar of while programs. $c$ is a constant location, $f$ is a field pointer, and $be$ is a background expression. In our logic,
we model every field $f$ as a function $f()$
from locations to the appropriate sort.}
\label{fig:prog-lang}
\end{figure}

\begin{figure}
\begin{align*}
\bot &\xRightarrow{\;\;*\;\;} \bot \\
(M, H, U ) &\xRightarrow{x:=y} ( M[x \mapsto y], H, U) \\
(M, H, U ) &\xRightarrow{x:=c} ( M[x \mapsto c], H, U) \\
(M, H, U ) &\xRightarrow{v:=be} ( M[v \mapsto be], H, U) \\
(M, H, U ) &\xRightarrow{x:=y.f}( M[x \mapsto f(y)], H, U) \text{, if $M(y)\in H$} \\
(M, H, U ) &\xRightarrow{x:=y.f} \bot \text{, if $M(y)\not\in H$} \\
(M, H, U ) &\xRightarrow{\mathsf{if}\; be\; \mathsf{then}\; S\; \mathsf{else}\; T} ( M', H', U') \text{, if $M \models be$ and $(M, H, U ) \xRightarrow{S} ( M', H', U')$} \\
(M, H, U ) &\xRightarrow{\mathsf{if}\; be\; \mathsf{then}\; S\; \mathsf{else}\; T} ( M', H', U') \text{, if $M \not\models be$ and $(M, H, U ) \xRightarrow{T} ( M', H', U')$} \\
(M, H, U ) &\xRightarrow{\mathsf{while}\; be\; \mathsf{do}\; S} ( M', H', U') \text{, if $M \models be$ and $(M, H, U ) \xRightarrow{S\; ;\; \mathsf{while}\; be\; \mathsf{do}\; S}( M', H', U')$} \\
(M, H, U ) &\xRightarrow{\mathsf{while}\; be\; \mathsf{do}\; S} (M, H, U ) \text{, if $M \not\models be$} \\
(M, H, U ) &\xRightarrow{x.f:=y}( M[f \mapsto f[M(x) \mapsto M(y)]], H, U) \text{, if $M(x)\in H$} \\
(M, H, U ) &\xRightarrow{x.f:=y} \bot \text{, if $M(x) \not\in H$} \\
\revII{(M, H, U )} &\revII{\xRightarrow{\mathsf{alloc}(x)}( M[x \mapsto a][f \mapsto f[a \mapsto \mathit{def}_f]], H \cup \{a\}, U \setminus \{a\}) \text{, for all $f \in F$}}  \\
& \hspace{5em} \revII{\text{for some $a \in U$ 
}} \\
( M[x \mapsto a], H, U ) &\xRightarrow{\mathsf{free}(x)}(M, H \setminus \{a\}, U) \text{, if $a \in H$} \\
(M[x \mapsto a], H, U ) &\xRightarrow{\mathsf{free}(x)} \bot \text{, if $a \not\in H$} \\
(M, H, U ) &\xRightarrow{S\; ;\; T}( M'', H'', U'') \text{, if $(M, H, U ) \xRightarrow{S} ( M', H', U')$} \\
\; & \hspace{8em} \text{and $( M', H', U') \xRightarrow{T} ( M'', H'', U'' )$} \\
\end{align*}
\caption{Operational Semantics of Frame Logic Programming Language}
\label{fig:op-sem}
\end{figure}

\subsection{Operational Semantics}\label{sec:op-sem}

In this Section we will discuss the operational semantics of our programs. First, we extend the signature of our logic by an infinite set $\Dummyloc = \{\dummyloc_i \mid i \in \N \}$ of constants of the type of the foreground sort, i.e. heap locations. We also add constants $\{\mathit{def}_f\}_{f \in F}$ to denote `default' values.

A configuration in our operational semantics is of the form $(M,H,U)$ where $M$ is a model that contains interpretations for the store and the heap. The store is a partial map that interprets variables, constants, and non-mutable functions over universes of the appropriate sorts. The heap is a total map on the domain of locations that interprets mutable functions. $H$ is a subset of (the universe of) locations denoting the set of allocated locations, and $U$ is a subset of locations denoting unallocated locations that can be allocated in the future. Lastly, we introduce a special configuration $\bot$ to denote an error state.

A configuration $(M,H,U)$ is \emph{valid} if:
\begin{itemize}
    \item all program variables of the location sort (including $nil$) map only to locations not in $U$.
    \item $U$ does not intersect with $H$,
    \item $U$ is infinite.
    \item the constants $\dummyloc_i$ are interpreted to distinct locations (i.e., no two constants in $\Dummyloc$ are interpreted to the same location) that are neither in $H$ nor in $U$. We will later use these constants to design our program logic rule for allocation.
    \item it is not possible to reach any location in $U$ or $\Dummyloc$ from any location in $H$. Simply, locations in $H$ do not point to locations in $U$ or $\Dummyloc$.
    \item locations in $U$ as well as those interpreted by $\Dummyloc$ have default values under functions, defined by the interpretation of the constants $\{\mathit{def}_f\}_{f \in F}$. 
\end{itemize}

We denote this by $\valid(M,H,U)$. We will demand that initial configurations of programs are valid, and maintain this by constructing our operational semantics rules to ensure that reachable configurations of programs are valid. Observe that validity of a configuration is expressible in First-Order Logic.




The full operational semantics are in Figure~\ref{fig:op-sem}. We abuse notation by using $M(x)$ to mean the value of $x$ as interpreted by the model $M$, as well as using the expression $M[x \mapsto y]$ to denote that the model is updated with the variable $x$ now storing the value stored by the variable/expression $y$ in the original model.

$\bot$ is a sink state, and every statement on $\bot$ transitions to $\bot$. The pointer lookup rule changes the store where the variable $x$ now maps to $f(y)$, provided $y$ is a location that is allocated. The pointer modification rule modifies the heap on the function $f$, where the store's interpretation for $x$ now maps to the store's interpretation for $y$ (again, provided $x$ is a location that is allocated). The allocation rule is the only nondeterministic rule in the operational semantics, as there is a transition for each $a \in U$. For each such $a$, the store is modified where $x$ now points to $a$. Additionally, the heap is modified for each function $f$ where the newly allocated $a$ maps to the default value under each $f$. Further, note since freed elements are not added back to $U$, freed locations cannot be reallocated. However, since $U$ is infinite, and dereferenced pointers must be in $H$, this does not pose a problem. All other rules are straightforward.

Note that when side conditions are violated as in the lookup rule or pointer modification rule, the configuration transitions to $\bot$, which denotes an abort or fault configuration. These faulting transitions are crucial for the soundness of the frame rule. \revII{The allocation rule, however, will always succeed for valid configurations because we demand that $U$ is infinite for valid configurations.}

\subsection{Triples and Validity}\label{sec:triples-validity}
We express specifications of programs using triples of the form $\{\alpha\}S\{\beta\}$ where $\alpha$ and $\beta$ are FL formulas in our extended signature and $S$ is a program in our while-language. \revII{However, we restrict the formulas that can appear in the specifications. First, we disallow atomic relations between locations. Second, we disallow functions from a background sort to the foreground sort (see Section~\ref{sec:frame-logic}). Finally, quantified formulas can have supports as large as the entire heap but we want our program logic to cover a more practical fragment without compromising expressive power. Thus, we require guards in quantification to be of the form $f(z') = z$ where $z$ is the quantified variable. We will maintain this as an invariant in the formulas that we generate for, say, weakest preconditions. Apart from these restrictions, we will also assume that formulas only feature unary functions for ease of presentation.}

We define a triple to be \emph{valid} if every valid configuration 
with heaplet being precisely the \spp{} of $\alpha$, when acted on by the program, yields a configuration with heaplet being the \spp{} of $\beta$.
More formally, a triple is valid if for every valid configuration $(M,H,U)$ such that $M \models \alpha$, $H = \sem[M]{\Fr(\alpha)}$: 

\noindent
(1) it is never the case that the abort state $\bot$ is encountered in the execution on $S$
    
\noindent
(2) if $(M,H,U)$ transitions to $(M',H',U')$ on $S$, then $M' \models \beta$ and $H' = \sem[M']{\Fr(\beta)}$

\revII{Note that the post-configuration $(M',H',U')$ will also be valid as a consequence of our transition rules, using the fact that $(M,H,U)$ is valid.}

\subsection{Program Logic}\label{sec:prog-log}

First, we define a set of \emph{local rules}
and rules for conditionals, while, sequence, consequence, and framing. \revII{Recall we assume there is at least one mutable function $f \in \Fm$ (see Section~\ref{sec:frame-logic}) in our signature.}
\begin{align*}
    \text{\textbf{Assignment:}} \;\; &\{\top\}\; x\,:=\, y\; \{x = y\} \hspace*{5ex} \{\top\}\; x\,:=\, c\; \{x = c\} \\
    \text{\textbf{Lookup:}} \;\; &\{f(y) = f(y)\}\; x\, :=\, y.f\; \{x = f(y)\} \\
    \text{\textbf{Mutation:}} \;\; &\{f(x) = f(x)\}\; x.f\, :=\, y\; \{f(x) = y\} \\
    \text{\textbf{Allocation:}} \;\; &\{\top\}\; \mathsf{alloc}(x)\; \{\bigwedge_{f \in F} f(x) = \mathit{def}_f \} \\
    \text{\textbf{Deallocation:}} \;\; &\{f(x) = f(x)\}\; \mathsf{free}(x)\; \{\top\} \\
    \text{\textbf{Conditional:}} \;\; &\hspace*{-23.9ex}
    \centerline{
    \infer{\{\alpha\}\; \mathsf{if}\; be\; \mathsf{then}\; S\; \mathsf{else}\; T\; \{\beta\}}{\{be \land \alpha\}\; S\; \{\beta\} & \{\neg be \land \alpha\}\; T\; \{\beta\}}} \\
    \text{\textbf{While:}} \;\; &\hspace*{-23.1ex} 
    \centerline{
    \infer{\{\alpha\}\; \mathsf{while}\; be\; \mathsf{do}\; S\; \{\neg be \land \alpha\}}{\{\alpha \land be\}\; S\; \{\alpha\}}} \\
    \text{\textbf{Sequence:}} \;\; &\hspace*{-23.1ex}
    \centerline{
    \infer{\{\alpha\}\; S\;;\; T\; \{\mu\}}{\{\alpha\}\; S\; \{\beta\} & \{\beta\}\; T\; \{\mu\}}} \\ 
    \text{\textbf{Consequence:}} \;\;&\hspace*{-17ex}
    \centerline{
    \infer{\{\alpha'\}\; S\; \{\beta'\}}{
    \begin{aligned} \alpha' \implies \alpha \\ \beta \implies \beta' \end{aligned} & \{\alpha\}\; S\; \{\beta\} &
    \begin{aligned} \Fr(\alpha) = \Fr(\alpha') \\ \Fr(\beta) = \Fr(\beta') \end{aligned}}} \\
    \text{\textbf{Frame:}} \;\;&\hspace*{-12ex}
    \centerline{
    \infer[\mathit{vars}(S) \cap \mathit{fv}(\mu) = \emptyset]{\{\alpha \wedge \mu\}\; S\; \{\beta \wedge \mu\}}{\Fr(\alpha) \cap \Fr(\mu) = \emptyset & \{\alpha\}\; S\; \{\beta\}}}
\end{align*}

The above rules are intuitively clear and are similar to the local rules in separation logic~\cite{reynolds02}. The rules for statements capture their semantics using minimal/tight heaplets,
and the frame rule allows proving triples with larger heaplets.
Some seemingly trivial preconditions as in the lookup, mutation, and allocation rules are added to ensure tight heaplets, i.e. that the support of the precondition is equal to the support of the postcondition, modulo any \textsf{alloc}/\textsf{free} statements. In the rule for $\mathsf{alloc}$, the postcondition says that the newly allocated
location has default values for all pointer fields and datafields (denoted as $\textit{def}_f$).
The soundness of the frame rule relies crucially on the frame theorem for FL (Theorem~\ref{thm:frametheorem}). The full soundness proof can be found in Section~\ref{sec:prog-log-proofs}.

\begin{theorem}\label{thm:soundness-local}
The above rules are sound with respect to the operational semantics.
\end{theorem}

\subsection{Weakest-Precondition Proof Rules}\label{sec:weakest-pre-rules}
We now turn to the much more complex problem of designing rules that give weakest preconditions for arbitrary postconditions, for loop-free
programs. In separation logic, such rules resort to using the magic wand operator $-*$~\cite{demri15,ohearn12,ohearn01,reynolds02}. The magic wand is a complex operator whose semantics calls for \emph{second-order quantification}~\cite{magic-wand-complexity} over arbitrarily large submodels. In our setting, our main goal is to show
that FL is itself capable of expressing weakest preconditions of postconditions written in FL.

First, we define a notion of \textit{Weakest Tightest Precondition} (WTP) of a formula $\beta$ with respect to each command that can figure in a basic block: assignment, lookup, mutation, allocation, and deallocation. To define this notion, we first define the notion of a preconfiguration:


\begin{definition}\label{def:precfg}
The \textit{preconfigurations} corresponding to a valid configuration $(M, H, U)$ with respect to a program $S$ are a set of valid configurations of the form $(M_p,H_p, U_p)$ such that when $S$ is executed on $M_p$ with unallocated set $U_p$ it dereferences only locations in $H_p$ and results (using the operational semantics rules) in $(M, H, U)$. That is:
\begin{align*}
\mathit{preconfigurations}&((M,H,U),S) = & \\
& \{(M_p,H_p,U_p) \,|\, \valid(M_p,H_p,U_p) \textrm{ and } (M_p,H_p,U_p) \overset{S}{\Rightarrow} (M,H,U) \} 
\end{align*}
\end{definition}

\begin{definition}\label{def:wtp}
We say that $\alpha$ is a \textit{Weakest Tightest Precondition} (WTP) of a formula $\beta$ with respect to a program $S$ if the set of all valid configurations that satisfy $\alpha$ is the same as the set of all preconfigurations of all valid configurations that satisfy $\beta$, with the addition of similar conditions on the allocated locations. More formally:
\begin{align*} 
\{ (M_p, H_p, U_p) \mid\, M_p \models \alpha,H_p= \sem[M_p]{\Fr(\alpha)}, \valid(M_p, H_p, U_p)\} &\\
=\{C \mid C \in \mathit{preconfigurations} (C_{post},S) & \textrm{ for some } C_{post} = (M,H,U) ,\\
 M \models \beta, &\, H = \sem[M]{\Fr(\beta)}, \valid(M,H,U)\}
\end{align*}

\end{definition}


With the notion of weakest tightest preconditions, we define global program logic rules for each command of our language. In contrast to local rules, global specifications contain heaplets that may be larger than the smallest heap on which one can execute the command.

Intuitively, a WTP of $\beta$ for lookup states that $\beta$ must hold in the precondition
when $x$ is interpreted as $x'$, where $x'=f(y)$, and further that the location
$y$ must belong to the \spp{} of $\beta$. The rules for mutation and allocation are more complex. For mutation, we define a transformation $\mwptr{x}{y}(\beta)$ that evaluates a formula $\beta$ in the pre-state as though it were evaluated in the post-state. We similarly define such a transformation $\mwalloc{x}$ for allocation.
We will define these in detail later. Finally, the deallocation rule ensures $x$ is not in the \spp{} of the postcondition. The conjunct $f(x) = f(x)$ is provided to satisfy the tightness condition, ensuring that the \spp{} of the precondition is the \spp{} of the postcondition with the addition of $\{x\}$. The rules can be seen below, and the proof of soundness for these global rules can be found in Section~\ref{sec:prog-log-proofs}.

\begin{align*}
    \text{\textbf{Assignment-G:}} \;\; &\{\beta[y/x]\}\; x\, :=\, y\; \{\beta\} \hspace*{5ex} \{\beta[c/x]\}\; x\, :=\, c\; \{\beta\} \\
    \text{\textbf{Lookup-G:}} \;\; &\{ \exists x':\; x' = f(y) .\; (\beta \land y \in \Fr(\beta))[x'/x] \}\; x\, :=\, y.f\; \{\beta\} \tag{where $x'$ does not occur in $\beta$} \\
    \text{\textbf{Mutation-G:}} \;\; &\{ \mwptr{x}{y}(\beta \land x \in \Fr(\beta))\} \;x.f\, :=\, y\; \{\beta\} \\
    \text{\textbf{Allocation-G:}} \;\; &\{\mwallocparam{\dummyloc_j}{x}(\beta)) \}\; \mathsf{alloc}(x)\; \{\beta\} \tag{for some $j$ where $\dummyloc_{j}$ does not occur in $\beta$} \\
    \text{\textbf{Deallocation-G:}} \;\; &\{ \beta \land x \not\in \Fr(\beta) \land f(x) = f(x) \} \; \mathsf{free}(x) \; \{\beta\} \tag{where $f \in \Fm$ is an arbitrary (unary) mutable function}
\end{align*}

\begin{theorem}
The rules above suffixed with \textbf{-G} are sound w.r.t the operational semantics. And, each precondition corresponds to the weakest tightest precondition of $\beta$.
\end{theorem}

\subsection{Definitions of \texorpdfstring{$MW$}{MW} Primitives: Mutation}
\label{sec:mwdefinitions}

Recall that the $MW$\footnote{The acronym MW is a shout-out to the Magic Wand operator, as 
these serve a similar function, except that
they are definable in FL itself.} 
primitives $\mwptr{x}{y}$ and $\mwalloc{x}$ need to evaluate a formula $\beta$ in the pre-state as it would evaluate in the post-state after mutation and allocation statements. The definition of $\mwptr{x}{y}$ is as follows:
\begin{align*}
\mwptr{x}{y}(\beta) = \beta[\lambda z.\; \ite(z=x : \ite(f(x) = f(x) : y, y), f(z))/f]
\end{align*}

\noindent The $\beta[\lambda z. \rho(z) / f]$ notation is shorthand for saying that each occurrence of a term of
the form $f(t)$, where $t$ is a term, is substituted (recursively, from inside out) by the term $\rho(t)$. The precondition essentially
evaluates $\beta$ taking into account $f$'s transformation, but we use the $ite$ expression with a tautological guard $f(x) = f(x)$ (which has the \spp{} containing the singleton $x$)
in order to preserve the \spp{} (see Section~\ref{sec:prog-log-proofs}: Lemma~\ref{lem:pointer-mod}). The definition of $\mwalloc{x}$ is similar, but involves a few subtleties. 


\subsection{Definitions of \texorpdfstring{$MW$}{MW} Primitives: Allocation}\label{sec:appendix-mwdefinitions}

\begin{figure*}
\[
  \begin{array}{rcl}
    \halloc{x}_c(\overline{y},z) &:=& \bot ~~~~\mbox{ for a constant $c$}\\\vspace*{1ex}
    \halloc{x}_w(\overline{y},z) &:=& \bot ~~~~\mbox{ for a variable $w$}\\
    \vspace*{1ex}
    \halloc{x}_{f(t)}(\overline{y},z) &:=& 
  \begin{cases}
  \left( z = \mwalloc{x}(t)\right) \lor  \halloc{x}_{t}(\overline{y},z) \land \left( \mwalloc{x}(f(t) = f(t))\right) & \mbox{if } f \in \Fm \\\vspace*{1ex}
  \halloc{x}_{t}(\overline{y},z) & \mbox{if } f \not\in \Fm 
  \end{cases}\\\vspace*{1ex}
  \halloc{x}_{\Fr(\beta)}(\overline{y},z) &:=& \halloc{x}_{\beta}(\overline{y},z)\\\vspace*{1ex}
 \halloc{x}_{(t_1 = t_2)}(\overline{y},z) &:=& \halloc{x}_{t_1}(\overline{y},z) \lor \halloc{x}_{t_2}(\overline{y},z) \\\vspace*{1ex}
 \revII{\halloc{x}_{(\top)}(\overline{y},z)} &\revII{:=}& \revII{\bot} \\\vspace*{1ex}
  \halloc{x}_{(\bot)}(\overline{y},z) &:=& \bot \\\vspace*{1ex}
 \halloc{x}_{R(\overline{t})}(\overline{y},z)  &:=& \halloc{x}_{\rho_{R(\overline{w})}}(\overline{y},z)[\mwalloc{x}(\overline{t})/\overline{w}] \lor \bigvee_{i=1}^n \halloc{x}_{t_i}(\overline{y},z)\\\vspace*{1ex}
 && \mbox{for $R \in \I$ with definition $R(\overline{w}) := \rho_R(\overline{w})$} \\\vspace*{1ex}
\revII{\halloc{x}_{\beta_1 \land \beta_2}(\overline{y},z)}  &:=& \revII{\halloc{x}_{\beta_1}(\overline{y},z) \lor \halloc{x}_{\beta_2}(\overline{y},z)}\\\vspace*{1ex}
\halloc{x}_{\lnot \beta}(\overline{y},z) &:=& \halloc{x}_{\beta}(\overline{y},z)\\\vspace*{1ex}
\halloc{x}_{\ite(\gamma:\beta_1,\beta_2)}(\overline{y},z)  &:=&
\halloc{x}_{\gamma}(\overline{y},z) \lor
\ite(\mwalloc{x}(\gamma):\halloc{x}_{\beta_1}(\overline{y},z),\halloc{x}_{\beta_2}(\overline{y},z)) \\\vspace*{1ex}
\halloc{x}_{\ite(\gamma:t_1,t_2)}(\overline{y},z)  &:=& \halloc{x}_{\gamma}(\overline{y},z)
\lor \ite(\mwalloc{x}(\gamma):\halloc{x}_{t_1}(\overline{y},z),\halloc{x}_{t_2}(\overline{y},z)) \\\vspace*{1ex}
\halloc{x}_{\exists w: \gamma.\beta}(\overline{y},z) &:=& \exists
w:\left(\mwalloc{x}(\gamma)\right).\left(\halloc{x}_{\gamma}(\overline{y},z) \lor \left(\mwalloc{x}(\gamma) \land \halloc{x}_{\beta}(\overline{y},z)\right)\right)
  \end{array}
\]
\caption{Definition of $\halloc{x}$ for use in $\mwalloc{x}$.}
\label{fig:halloc}
\end{figure*}

We have already seen the definition of $\mwptr{x}{y}$ in Section~\ref{sec:mwdefinitions}. Observe that the inner guard contains the expression $f(x) = f(x)$, which is a tautology that only serves to include $x$ in the \spp{} of the transformed formula (as is required by the weakest tightest precondition definition). This is similar to the Separation Logic syntax $x \mapsto \text{\textunderscore}$. Of course, one can use syntax sugars such as $\mathit{acc}(x)$ available in tools like Viper~\cite{vipertool} to make the intent clearer.

We will detail the construction of $\mwalloc{x}$ in this section.

$\mwalloc{x}$, like $\mwptr{x}{y}$, is also meant to evaluate a formula in the pre-state as though it were evaluated in the post-state. However, note that the \spp{} of this formula must not contain the allocated location (say $x$). Since we know from the operational semantics of allocation that the allocated location is going to point to default values, we can proceed similarly as we did for the previous definition, identify terms evaluating to $f(x)$ and replace them with the default value (under $f$). This has the intended effect of evaluating to the same value as in the post-state while removing $x$ from the \spp{}. 

However, this approach fails when we apply it to \spp{} expressions (since removing $x$ from the \spp{} guarantees that we would no longer compute the `same' value as of that in the post-state). In particular, a subformula of the form $t \in \Fr(\gamma)$ may be falsified by that transformation. To handle this, we identify when $x$ might be in the \spp{} of a given expression and replace it with $v$ (which is given as a parameter) such that neither $x$ nor $v$ is dereferenced, and will not be in the \spp{} of the resulting transformation. We do this syntactic replacement of $x$ with $v$ for formulas not within \spp{} expressions as well. This has the effect of eliminating $x$ altogether from the formula.

We define $\mwalloc{x}$ inductively. We first consider the case where $\beta$ does not contain any subformulas involving \spp{} expressions or inductive definitions. Then, we have $\mwalloc{x}$ defined as follows:
\[
\mwalloc{x}(\beta) =
\beta[v/x][\lambda z.\; \ite(z=v : \mathit{def}_f, f(z))/f]_{f \in F}
\]

\noindent where this means for each instance of a (mutable or immutable) function $f$ in $\beta$, we replace $f(x)$ with a default value. We also replace all free instances of $x$ in $\beta$ with $v$.

If $\Fr(\gamma)$ is a subterm of $\beta$, we translate it to a formula $\halloc{x}_\gamma$ inductively as in Figure~\ref{fig:halloc}. This definition is very similar to the translation of FL formulas to FO-RD in Figure~\ref{fig:translation-to-ford} where we replace free instances of $x$ with $v$. Since this is a relation, we must transform membership to evaluation, i.e, transform expressions of the form $t \in \Fr(\gamma)$ to $\halloc{x}_\gamma(\overline{y},\mwalloc{x}(t))$ where $\overline{y}$ are the free variables (we transform inductively--- at the highest level free variables will be program/ghost variables). We also transform union of \spp{} expressions to disjunction of the corresponding relations, equality to (quantified) double implication, etc.

For a subterm of $\beta$ of the form $I(\overline{t})$ where $I$ is an inductive definition with body $\rho_I$, we translate it to $I'(\mwalloc{x}(\overline{t}))$ where the body of $I'$ is defined as $\mwalloc{x}(\rho_{I'})$.

The above cases can be combined with boolean operators and if-then-else, which $\mwalloc{x}$ distributes over.

Lastly, to design the program logic rule we have to decide the value of the parameter $v$. The idea is that $v$ (which is essentially $x$) will hold the location that is going to be allocated. Since we do not know which one of the locations in the unallocated set will be actually allocated next, and $v$ will not feature in the support of the $\mwalloc{x}$ formula by construction, it is enough to choose an element that simply mirrors the behaviour of the location to be allocated, in that it will evaluate the same way under any function and is different from any location that has ever been allocated. This is where we use our `dummy' location constants $\dummyloc_i$ with which we extended our signature (see Section~\ref{sec:op-sem}). We also maintain the invariant that these dummy locations will not feature in the support of any preconditions that we generate. Note that we ensure that these dummy locations are indeed different by first demanding that the constants bear distinct values (see Section~\ref{sec:op-sem}) and ensuring that we use a constant that does not appear in the given postcondition (see Section~\ref{sec:weakest-pre-rules}).

\subsection{Program Logic Proofs}\label{sec:prog-log-proofs}

In this section we present soundness proofs for the global rules developed in Section~\ref{sec:weakest-pre-rules}. We often drop $U$ when referring to a configuration $(M,H,U)$ for ease of presentation since $U$ is only modified by the allocation rule.
\begin{theorem}[Lookup Soundness]\label{thm:lookup-soundness}
Let $M$ be a model and $H$ a sub-universe of locations such that 
\begin{gather*} 
M \models \exists x' : x' = f(y).\; (\beta \land y \in \Fr(\beta))[x'/x] \\
H = \sem[M]{\Fr(\exists x' : x' = f(y).\; (\beta \land y \in \Fr(\beta)))[x'/x]\,}
\end{gather*}
Then $(M,H) \xRightarrow{x:=y.f} (M',H')$, $M' \models \beta \text{, and } H' = \sem[M']{\Fr(\beta)}$.
\end{theorem}
\begin{proof}
Observe that $\sem[M]y \in H$ since $y$ is in the \spp{} of the precondition. Therefore we know $(M,H) \xRightarrow{x:=y.f} (M',H')$ where $M' = M[x \mapsto \sem[M]{f(y)}]$ and $H' = H$. Next, note that if there is a formula $\alpha$ (or term $t$) where $x$ is not a free variable of $\alpha$ (or $t$), then $M$ and $M'$ have the same valuation of $\alpha$ (or $t$). This is true because the semantics of lookup only changes the valuation for $x$ on $M$. In particular, $M' \models \exists x' : x' = f(y).\; (\beta \land y \in \Fr(\beta))[x'/x]$. Thus,
\begin{align*}
&M' \models \exists x' : x' = f(y).\; (\beta \land y \in \Fr(\beta))[x'/x] \\
&\implies M'[x' \mapsto c] \\
&\hspace*{7ex} \models x' = f(y) \land (\beta \land y \in \Fr(\beta))[x'/x] \tag{for some $c$} \\
&\implies M'[x' \mapsto \sem[M']{f(y)}] \models (\beta \land y \in \Fr(\beta))[x'/x] \tag{since $f$ is a function} \\
&\implies M'[x' \mapsto \sem[M']{x}] \models (\beta \land y \in \Fr(\beta))[x'/x] \tag{operational semantics} \\
&\implies M'[x' \mapsto \sem[M']{x}] \\
&\hspace*{5ex} \models ((\beta \land y \in \Fr(\beta))[x'/x])[x/x'] \\
&\implies M'[x' \mapsto \sem[M']{x}] \models \beta \land y \in \Fr(\beta) \\
&\implies M' \models \beta \land y \in \Fr(\beta) \tag{$\beta$ does not mention $x'$} \\
&\implies M' \models \beta
\end{align*}
The heaplet condition follows from a similar argument. Specifically
\allowdisplaybreaks
\begin{align*}
&H' = H \\
&= \sem[M]{\Fr(\exists x' : x' = f(y).\; (\beta \land y \in \Fr(\beta)))[x'/x]\,} \\
&= \sem[M']{\Fr(\exists x' : x' = f(y).\; (\beta \land y \in \Fr(\beta)))[x'/x]\,} \tag{does not mention $x$} \\
&= \{\sem[M']{y}\} \cup \sem[M']{\Fr((\beta \land y \in \Fr(\beta))[x'/x])}(x' \mapsto \sem[M']{f(y)}) \tag{def of $\Fr$} \\
&= \{\sem[M']{y}\} \cup \sem[M']{\Fr((\beta \land y \in \Fr(\beta))[x'/x])}(x' \mapsto \sem[M']{x}) \tag{operational semantics} \\
&= \{\sem[M']{y}\} \cup \sem[M']{\Fr(\beta \land y \in \Fr(\beta))} \tag{similar reasoning as above} \\
&= \sem[M']{\Fr(\beta)} \tag{since $M' \models y \in \Fr(\beta)$ from above}
\end{align*}
\end{proof}

\begin{theorem}[WTP Lookup]\label{thm:lookup-weakness}
Let $M,M'$ be models with $H,H'$ sub-universes of locations (respectively) such that $(M,H) \xRightarrow{x := y.f} (M',H')$, $M' \models \beta$ and $H' = \sem[M']{\Fr(\beta)}$. Then
\begin{align*}
&M \models \exists x' : x' = f(y).\; (\beta \land y \in \Fr(\beta))[x'/x] \tag{weakest-pre} \\
&H = \sem[M]{\Fr(\exists x' : x' = f(y).\; (\beta \land y \in \Fr(\beta)))} \tag{tightest-pre}
\end{align*}
\end{theorem}
\begin{proof}
Both parts follow by simply retracing steps in the above proof. The weakness claim follows from the first part of the proof above, where all implications can be made bidirectional (using operational semantics rules, definition of existential quantifier, etc.). The tightness claim follows immediately from the second part of the proof above as all steps involve equalities.
\end{proof}

For soundness of the pointer modification rules, we prove the following lemma:

\begin{lemma}\label{lem:pointer-mod}
Given a formula $\beta$ (term $t$) and configurations $(M,H)$ and $(M',H')$ such that $(M,H)$ transforms to $(M',H')$ on the command $x.f := y$, then $\sem[M]{\mwptr{x}{y}(\beta)} = \sem[M']{\beta}$. Additionally, $\sem[M]{\Fr(\mwptr{x}{y}(\beta))} = \sem[M']{\Fr(\beta)}$. Both equalities hold for terms $t$ as well.
\end{lemma}
\begin{proof}
Induction on the structure of $\beta$, unfolding $\mathit{MW}^{x.f:=y}(\beta)$ accordingly. We discuss one interesting case here, namely when $\beta$ has a subterm of the form $f(t)$. Now, we have two cases, depending on whether $\sem[M]{\mwptr{x}{y}}(t) = \sem[M]{x}$. If it does, then
\begin{align*}
&\sem[M]{\mwptr{x}{y}(f(t))} \\ 
&= \sem[M]{\ite(\mwptr{x}{y}(t):\ite(f(x)=f(x):y,y),f(\mwptr{x}{y}(t)))} \tag{definition} \\
&= \sem[M]{\ite(f(x)=f(x):y,y)} \tag{assumption} \\
&= \sem[M]{y}  \\
&= \sem[M']{y} \\
&= \sem[M']{f}(\sem[M']{x}) \tag{def of $f$ on $M'$} \\
&= \sem[M']{f}(\sem[M]{x}) \tag{definition of $M,M'$} \\
&= \sem[M']{f}(\sem[M]{\mwptr{x}{y}(t)}) \tag{assumption} \\
&= \sem[M']{f}(\sem[M']{t}) \tag{induction hypothesis} \\
&= \sem[M']{f(t)} \tag{definition}
\end{align*}
The proof for the cases when $\sem[M]{\mwptr{x}{y}(t)} \neq \sem[M]{x}$ and the heaplet equality claims are similar, and all other cases are trivial.
\end{proof}

\begin{theorem}[Mutation Soundness]\label{thm:pointer-soundness}
Let $M$ be a model and $H$ a sub-universe of locations such that 
\begin{gather*}
M \models \mwptr{x}{y}(\beta \land x \in \Fr(\beta)) \\
H = \sem[M]{\Fr(\mwptr{x}{y}(\beta \land x \in \Fr(\beta)))}
\end{gather*}
Then $(M,H) \xRightarrow{x.f:=y} (M',H'), M' \models \beta \text{, and } H' = \sem[M']{\Fr(\beta)}$
\end{theorem}
\begin{proof}
From the definition of the transformation $\mwptr{x}{y}$, we have that $\mwptr{x}{y}(\beta \land x \in \Fr(\beta))$ will be transformed to the same formula as $\mwptr{x}{y}(\beta) \land x \in \Fr(\mwptr{x}{y}(\beta))$, the heaplet of which, since the formula holds on $M$, contains $x$. Therefore $x \in H$ and from the operational semantics we have that $(M,H) \xRightarrow{x.f:=y} (M',H')$ for some $(M',H')$ such that $H = H'$.

From Lemma~\ref{lem:pointer-mod} we have that $M'\models \beta \land x \in \Fr(\beta)$, since $M$ models the same. In particular $M' \models \beta$. Moreover we have 
\begin{align*}
H' 
& = H \tag{operational semantics}\\
& = \sem[M]{\Fr(\mwptr{x}{y}(\beta \land x \in \Fr(\beta)))} \tag{given}\\
& = \sem[M']{\Fr(\beta \land x \in \Fr(\beta))} \tag{Lemma~\ref{lem:pointer-mod}}\\
& = \sem[M']{\Fr(\beta)} \tag{semantics of H operator}
\end{align*}

Therefore $M' \models \beta$ and $H' = \sem[M']{\Fr(\beta)}$ which makes our pointer mutation rule sound.
\end{proof}

\begin{theorem}[WTP Mutation]\label{thm:pointer-weakness}
Let $M,M'$ be models with $H,H'$ sub-universes of locations (respectively) such that $(M,H) \xRightarrow{x.f := y} (M',H')$, $M' \models \beta$ and $H' = \sem[M']{\Fr(\beta)}$. Then
\begin{align*}
&M \models \mwptr{x}{y}(\beta \land x \in \Fr(\beta)) \tag{weakest-pre} \\
&H = \sem[M]{\Fr(\mwptr{x}{y}(\beta \land x \in \Fr(\beta)))} \tag{tightest-pre} \end{align*}
\end{theorem}
\begin{proof}
From the operational semantics, we have that $(M,H) \xRightarrow{x.f := y} (M',H')$ only if $x \in H$ and $H = H'$. Therefore $x \in H' = \sem[M']{\Fr(\beta)}$ (given) which in turn implies that $M' \models \beta \land x \in \Fr(\beta)$ as well as $H' = \sem[M']{\Fr(\beta \land x \in \Fr(\beta))}$. Applying Lemma~\ref{lem:pointer-mod} yields the result. 
\end{proof}

\begin{lemma}\label{lem:halloc-meaning}
Given a formula $\beta$ (or term $t$) and configurations $(M,H,U)$ and $(M',H',U')$ such that $(M,H,U)\xRightarrow{\mathsf{alloc}(x)} (M',H',U')$, it is the case that $\sem[{M[v \mapsto a]}]{\halloc{x}(\overline{y},\mwalloc{x}(t))}$ iff $\sem[M']{t \in \Fr(\beta)}$, where $a = \sem[M']{x}$ and $\overline{y}$ are the free variables in $\mwalloc{x}(\beta)$. Additionally, $\sem[{M[v \mapsto a]}]{\Fr(\halloc{x}_{\beta}(\overline{y},z))} = \sem[M']{\Fr(\beta) \setminus \{x\}}$ where $z$ is a free variable. Both equalities hold for terms $t$ as well.
\end{lemma}
\begin{proof}
Induction on the structure of $\beta$ and using the construction in Figure~\ref{fig:halloc}. For the second claim about the \spp{} of $\halloc{x}$, the fact that we only allow specific kinds of guards is crucial in the inductive case of the existential quantifier.
\end{proof}

\begin{lemma}\label{lem:alloc-soundness-lemma}
Given a formula $\beta$ (or term $t$) and configurations $(M,H,U)$ and $(M',H',U')$ such that $(M,H,U)\xRightarrow{\mathsf{alloc}(x)} (M',H',U')$, we have that $\sem[{M[v \mapsto a]}]{\mwalloc{x}(\beta)} = \sem[M']{\beta}$, where $a = \sem[M']{x}$. Additionally, $\sem[{M[v \mapsto a]}]{\Fr(\mwalloc{x}(\beta))} = \sem[M']{\Fr(\beta) \setminus \{x\}}$. Both equalities hold for terms $t$ as well.
\end{lemma}
\begin{proof}
First, we split on the structure of $\beta$, as the definition of $\mwalloc{x}$ differs depending on the form of $\beta$. For subformulas with no \spp{} expressions or inductive definitions, the proof follows from the syntactic definition of $\mwalloc{x}$ and is very similar to Lemma~\ref{lem:pointer-mod}. It is important to note that $a$ is a location different from any of the locations ever allocated, and is therefore different from any location held in any program variable or reachable by any recursive definition. This is crucial in the case of handling the atomic equality/disequality formulas.

Subformulas with \spp{} expressions follow by construction using Lemma~\ref{lem:halloc-meaning}, and formulas with inductive definitions follow by construction as well. Boolean combinations and if-then-else follow using the inductive hypothesis.
\end{proof}

We are now ready to prove the soundness and WTP property of the allocation rule. This will be slightly different from the other soundness theorems because it reasons only about configurations reachable by a program or a valid initial state. This strengthening of the premise is not an issue since we will only ever execute commands on such states. We shall first prove a lemma.

\begin{lemma}
\label{lem:agnostic-to-v}
Let $\beta$ be any formula within our restricted fragment (Section~\ref{sec:triples-validity}) and $(M,H,U)$ be a valid configuration. Then, for any locations $a_1,a_2 \in U \cup \Dummyloc$:
\begin{align*}
&\sem[{M[v \mapsto a_1]}]{\mwalloc{x}(\beta)}= \sem[{M[v \mapsto a_2]}]{\mwalloc{x}(\beta)}\\
&\hspace{25ex}\textrm{and} \\
&\sem[{M[v \mapsto a_1]}]{\Fr(\mwalloc{x}(\beta))} =
\sem[{M[v \mapsto a_2]}]{\Fr(\mwalloc{x}(\beta))}
\end{align*}
\end{lemma}
\begin{proof}
The proof follows by a simple inductive argument on the structure of $\beta$. First observe that in any model if $v$ is interpreted to a hitherto unallocated location (either from $U$ or one of the dummy constants) it is never contained in $\mwalloc{x}(\beta)$ since it is never dereferenced. Therefore, all we are left to prove is that the actual value of $v$ (between choices in $U \cup \Dummyloc$) influences neither the truth value nor the \spp{} of the formula. The key case is that of $ite$ expressions where the value of $v$ can influence the truth of the guard. This case can be resolved using the observation that since $(M,H,U)$ is a valid configuration, the value of any unallocated location can never equal that of a program variable. Since we have no atomic relations either in our restricted fragment, any two values in $U \cup \Dummyloc$ are indistinguishable by a formula in this fragment.

In particular, any $ite$ expressions that depend on the value of $v$ either compare it with a term over a program variable ---which is never equal, or compare it with a quantified variable --- which itself only takes on values allowed by the guard of the quantification that, inductively, does not distinguish between values in $U \cup \Dummyloc$.
\end{proof}

\begin{theorem}[Allocation Soundness]\label{thm:alloc-soundness}
Let $(M,H,U)$ be a valid configuration such that 
\begin{gather*}
M \models \mwallocparam{\dummyloc_j}{x}(\beta)
\\
H = \sem[M]{\Fr(\mwallocparam{\dummyloc_j}{x}(\beta))}
\end{gather*}
such that $\dummyloc_j$ does not appear in $\beta$. Then $(M,H,U) \xRightarrow{\mathsf{alloc}(x)} (M',H',U\setminus \sem[M']{x})$, $M' \models \beta \text{ and } H' = \sem[M']{\Fr(\beta)}$
\end{theorem}
\begin{proof}
It is easy to see that $(M,H,U) \xRightarrow{\mathsf{alloc}(x)} (M',H',U\setminus \sem[M']{x})$ for some $M', H'$ since $U$ is infinite for any valid configuration.

Let $a$ be the actual location allocated, i.e., $a = \sem[M']{x}$. Clearly $a \in U$ by the operational semantics. Then, we have:
\begin{align*}
&M \models \mwallocparam{\dummyloc_j}{x}(\beta)\\
& \implies M[v\mapsto \dummyloc_j] \models \mwalloc{x}(\beta)\\
& \implies M[v\mapsto a] \models \mwalloc{x}(\beta) \tag{Lemma~\ref{lem:agnostic-to-v}}\\
& \implies M' \models \beta \tag{Lemma~\ref{lem:alloc-soundness-lemma}}
\end{align*}

For the \spp{} claim, we have:
\begin{align*}
H &= \sem[M]{\Fr(\mwallocparam{\dummyloc_j}{x}(\beta))}\\
&= \sem[{M[v \,\mapsto\, \dummyloc_j]}]{\Fr(\mwalloc{x}(\beta))}\\
&= \sem[{M[v \,\mapsto\, a]}]{\Fr(\mwalloc{x}(\beta))} \tag{Lemma~\ref{lem:agnostic-to-v}} \\
&= \sem[M']{\Fr(\beta) \setminus\{x\}} \tag{Lemma~\ref{lem:alloc-soundness-lemma}}
\end{align*}
Now $H' = H \cup \{\sem[M']{x}\}$ (by operational semantics) $= \sem[M']{\Fr(\beta) \setminus \{x\}} \,\cup\, \{\sem[M']{x}\} = \sem[M']{\Fr(\beta)}$, as desired.
\end{proof}

\begin{theorem}[WTP Allocation]\label{thm:alloc-weakness}
Let $(M,H,U)$ and $(M',H',U\setminus \sem[M']{x})$ be valid configurations such that
\begin{align*}
&(M,H,U) \xRightarrow{\mathsf{alloc}(x)} (M',H',U\setminus \sem[M']{x})\\
&M' \models \beta \textrm{ and } H' = \sem[M']{\Fr(\beta)}
\end{align*}
Then
\begin{align*}
M &\models \mwallocparam{\dummyloc_j}{x}(\beta)\\
H &= \sem[M]{\Fr(\mwallocparam{\dummyloc_j}{x}(\beta))}\\
&\textrm{for some $j$ such that $\dummyloc_j$ does not appear in $\beta$.}
\end{align*}
\end{theorem}
\begin{proof}
The first claim follows easily from an application of Lemma~\ref{lem:alloc-soundness-lemma} followed by an application of Lemma~\ref{lem:agnostic-to-v}. For the second claim, observe that as done in the proof above for Theorem~\ref{thm:alloc-soundness} we can prove that $\sem[M']{\Fr(\beta)\setminus\{x\}} = \sem[M]{\Fr(\mwallocparam{\dummyloc_j}{x}(\beta))}$. The proof concludes by observing that by the operational semantics we have $H = H' \setminus \{\sem[M']{x}\} = \sem[M']{\Fr(\beta)} \setminus \{\sem[M']{x}\} = \sem[M']{\Fr(\beta) \setminus \{x\}}$
\end{proof}

\begin{theorem}[Deallocation Soundness]\label{thm:dealloc-soundness}
Let $M$ be a model and $H$ a sub-universe of locations such that 
\begin{gather*}
M \models \beta \land x \not\in \Fr(\beta) \land f(x) = f(x) \\
H = \sem[M]{\Fr(\beta \land x \not\in \Fr(\beta) \land f(x) = f(x))}
\end{gather*}
Then $(M,H) \xRightarrow{\mathsf{free}(x)} (M',H'), M' \models \beta \text{, and } H' = \sem[M']{\Fr(\beta)}$
\end{theorem}
\begin{proof}
Observe that $x \in \Fr(\beta \land x \notin \Fr\beta) \land f(x) = f(x))$, i.e., $\sem[M]{x} \in H$. Therefore we have from the operational semantics that $(M,H) \xRightarrow{\mathsf{free}(x)} (M',H')$ such that $M' = M$ and $H' = H \setminus \{\sem[M]{x}\}$. Since $M \models \beta \land x \not\in \Fr(\beta) \land f(x) = f(x)$, we know $M \models \beta$, which implies $M' \models \beta$.
Similarly, we have: 
\begin{align*}
H' &= H \setminus \{\sem[M]{x}\} \tag{operational semantics} \\
&= \sem[M]{\Fr(\beta \land x \not\in \Fr(\beta) \land f(x) = f(x))} \setminus \{\sem[M]x\} \\
&= \sem[M]{\Fr(\beta)} \cup \{\sem[M]x\} \setminus \{\sem[M]x\} \tag{def of $\Fr$} \\
&= \sem[M]{\Fr(\beta)} \\
&= \sem[M']{\Fr(\beta)} \tag{$M = M'$}
\end{align*}
\end{proof}

\begin{theorem}[WTP Deallocation]\label{thm:dealloc-weakness}
Let $M,M'$ be models with $H,H'$ sub-universes of locations (respectively) such that $(M,H) \xRightarrow{\mathsf{free}(x)} (M',H')$, $M' \models \beta$ and $H' = \sem[M']{\Fr(\beta)}$. Then
\begin{align*}
&M \models \beta \land x \not\in \Fr(\beta) \land f(x) = f(x) \tag{weakest-pre} \\
&H = \sem[M]{\Fr(\beta \land x \not\in \Fr(\beta) \land f(x) = f(x))} \tag{tightest-pre} \end{align*}
\end{theorem}
\begin{proof}
For the first part, note that the operational semantics ensures $\sem[M']x \not\in H' = \sem[M']{\Fr(\beta)}$ and $M = M'$. So $M' \models x \not\in \Fr(\beta)$ which implies $M \models x \not\in \Fr(\beta)$. Similarly, $M \models \beta$, and $M \models f(x) = f(x)$ as it is a tautology. Tightness follows from similar arguments as in Theorem~\ref{thm:dealloc-soundness}, again noting that $H' = H \setminus \{\sem[M]x\}$ as per the operational semantics.
\end{proof}

\begin{reptheorem}{thm:soundness-local}
The four local rules (for assignment, lookup, mutation, allocation, and deallocation) given in Section~\ref{sec:prog-and-prfs} are sound given the global rules.
\end{reptheorem}
\begin{proof}
The validity of assignment follows immediately setting $\beta$ to be $x = y$ (or $x = c$). Instantiating with this and the precondition becomes $y = y$ which is equivalent to $\mathit{true}$ (the heaplet of both is empty)

The validity of the next (lookup) follows since
\begin{align*}
&\wtp(x = f(y), x := y.f) \\
&= \exists x' : x' = f(y).\; (x=f(y) \land y \in \Fr(x=f(y)))[x'/x] \\
&= \exists x' : x' = f(y).\; x'=f(y) \land y \in \Fr(x'=f(y)) \\
&= f(y) = f(y) \land y \in \Fr(f(y) = f(y))
\end{align*}
This is a tautology, so it is clearly implied by any precondition, in particular the precondition $f(y) = f(y)$. Similarly, the \spp{} of the resulting formula is the singleton $\{y\}$ which is also the \spp{} of $f(y) = f(y)$ as needed.

For the second local rule (mutation), we first notice that
\begin{align*}
\mwptr{x}{y}(f(x) = y) &= (f(x) = y) [\ite(z = x: \ite(f(x)=f(x): y, y), f(z))/f(x)] \\
&= \ite(x = x: \ite(f(x)=f(x): y, y), f(x)) = y
\end{align*}
Then,
\begin{align*}
\wtp(f(x) = y, x.f := y) &= \ite(x = x: \ite(f(x)=f(x): y, y), f(x)) = y \\
&\land x \in \Fr(\ite(x = x: \ite(f(x)=f(x): y, y), f(x)) = y)
\end{align*}
The first conjunct is clearly true since it is equivalent to $y = y$. The second conjunct is also true because $\Fr(\ite(x = x: \ite(f(x)=f(x): y, y), f(x)) = y) = \{x\}$. Thus, this formula is also a tautology, and it is implied by the precondition $f(x) = f(x)$. Additionally the \spp{} of the resulting formula and the \spp{} of $f(x) = f(x)$ is $\{x\}$ as needed.

For the next local rule (allocation), observe that the postcondition does not have any \spp{} expressions or inductive definitions. Therefore, we have that:
\[ \mwalloc{x}(f(x) = \mathit{def}_f)
= \ite(x = x: \mathit{def}_f,f(x)) = \mathit{def}_f
\]
Observe that the \spp{} of the above expression is $\emptyset$. The \spp{} of a conjunction of such expressions is also $\emptyset$. This and the fact that $\mwalloc{x}$ distributes over $\land$ gives us:
\begin{align*}
& \wtp\left(\bigwedge\limits_{f \in F}\left(f(x) = \mathit{def}_f\right), x := alloc()\right)\\
&= \forall v:\, v \notin \emptyset \implies \mwalloc{x}\left(\bigwedge\limits_{f \in F}f(x) = \mathit{def}_f\right)\\
&= \forall v:\, \bigwedge\limits_{f \in F}\left(\ite(x = x: \mathit{def}_f,f(x)) = \mathit{def}_f\right)
\end{align*}
which is a tautology (as it is equivalent to $\mathit{def}_f = \mathit{def}_f$) and its \spp{} is $\emptyset$ as desired.

Finally the last local rule (deallocation) follows directly from the global rule for deallocation by setting $\beta = \top$.
\end{proof}

\begin{theorem}[Conditional, While Soundness] \label{thm:condwhile-soundness}
\end{theorem}
\begin{proof}
See any classical proof of the soundness of these rules, as in~\cite{apt81}.
\end{proof}

\begin{theorem}[Sequence Soundness]
\label{thm:soundness-seq}
The Sequence rule is sound.
\end{theorem}
\begin{proof}
Follows directly from the operational semantics.
\end{proof}

\begin{theorem}[Consequence Soundness]
\label{thm:soundness-conseq}
The Consequence rule is sound.
\end{theorem}
\begin{proof}
First, note if we can't execute $S$ then the triple is vacuously valid. Next, assume $M \models \alpha'$. Then, because $\alpha' \implies \alpha$, we know $M \models \alpha$. So, if we execute $S$ and result in $M'$, we know $M' \models \beta$ since $\{\alpha\}S\{\beta\}$ is a valid triple. Then, $M' \models \beta'$ since $\beta \implies \beta'$. Finally, since the \spps{} of $\alpha$ and $\alpha'$ as well as $\beta$ and $\beta'$ are equal, the validity of the Hoare triple holds.
\end{proof}

\begin{theorem}[Frame Rule Soundness]
\label{thm:soundness-frame}
The Frame rule is sound.
\end{theorem}
\begin{proof}
First, we establish that for any $(M,H)$ such that $M \models \alpha \land \mu$ and $H = \sem[M]{\Fr(\alpha \land \mu)}$ we never reach $\bot$. Consider $(M,H) \overset{S}{\Rightarrow^*}$ as the sequence of configurations $P_2$.
Construct the sequence of configurations $P_1$ as
$(M,\sem[M]{\Fr(\alpha)}) \overset{S}{\Rightarrow^*}$, where each allocation from $S$ in $P_1$ chooses the same location to allocate as in $P_2$.
We can show that for each step in $P_2$, there exists a corresponding step in $P_1$ such that:
\begin{enumerate}
    \item at any corresponding step the allocated set on $P_2$ is a superset of the allocated set on $P_1$
    \item the executions allocate and deallocate the same locations
\end{enumerate}

The claim as well as the first item is easy to show by structural induction on the program. Given that, the second is trivial since a location available to allocate on $P_2$ is also available to allocate on $P_1$. Any location that is deallocated on $P_2$ that is unavailable on $P_1$ would cause $P_1$ to reach $\bot$ which is disallowed since we are given that $\{\alpha\}S\{\beta\}$ is valid.

Thus if we abort on the former we must abort on the latter, which is a contradiction since we are given that $\{\alpha\}S\{\beta\}$ is valid. From the second item above, we can also establish that all mutations of the model are outside of $\Fr(\mu)$ since it is unavailable on $P_1$ (we start with $\Fr(\alpha)$ and allocate only outside $\Fr(\alpha \land \mu) = \Fr(\alpha) \cup \Fr(\mu)$, and we are also given that the \spps{} of $\alpha$ and $\mu$ are disjoint in any model). Therefore, if there exists a configuration $(M',H')$ such that $(M,H) \overset{S}{\Rightarrow^*} (M',H')$ it must be the case that $M'$ is a mutation of $M$ that is stable on $\Fr(\mu)$. Since $\{\alpha\}S\{\beta\}$ is valid we have that $M' \models \beta$. Lastly, we conclude from the Frame Theorem (Theorem~\ref{thm:frametheorem}) that since $M \models \mu$, $M' \models \mu$ which gives us $M' \models \beta \land \mu$.

We must also show that $H' = \sem[M']{\Fr(\beta \land \mu)}$. To show this, we can strengthen the inductive invariant above with the fact that at any corresponding step the allocated set on $P_2$ is not simply a superset of that on $P_1$, but in fact differs exactly by $\Fr(\mu)$. This invariant establishes the desired claim, which concludes the proof of the frame rule.
\end{proof}

\subsection{Example}
In this section, we will see an example of using our program logic rules that we described earlier. This will demonstrate the utility of Frame Logic as a logic for annotating and reasoning with heap manipulating programs, as well as offer some intuition about how our program logic can be deployed in a practical setting. The following program performs in-place reversal of the linked list pointed to by $i$
\begin{flalign*}
    &\revII{\texttt{j := nil ;}}&\\[-0.5em]
    &\revII{\texttt{while (i != nil) do }}&\\[-0.5em]
    &\hspace{3ex}\revII{\texttt{k := i.next ;}}&\\[-0.5em]
    &\hspace{3ex}\revII{\texttt{i.next := j ;}}&\\[-0.5em]
    &\hspace{3ex}\revII{\texttt{j := i ;}}&\\[-0.5em]
    &\hspace{3ex}\revII{\texttt{i := k}}&
\end{flalign*}

\revII{Although we must show that $j$ is the reverse of $i$ for full functional correctness, we illustrate our program logic using a simpler contract, namely that $j$ points to a $\lst$ and the end of the program. Full functional correctness can be proved similarly, modeling the content of the linked lists as mathematical sequences.} The recursive definition of $\lst$ we use for this proof is the one from Figure~\ref{fig:rec-defs-list}:
\begin{align*}
    \lst(x) :=\; &\ite(x = nil, \top, \exists z : z = \nxt(x).\; \lst(z) \land x \not\in \Fr(\lst(z)))
\end{align*}

We need to also give an invariant for the while loop, simply stating that $i$ and $j$ point to disjoint lists:
$\lst(i) \land \lst(j) \land \Fr(\lst(i)) \cap \Fr(\lst(j)) = \emptyset$.

We prove that this is indeed an invariant of the while loop below. Our proof uses a mix of both local and global rules from Sections~\ref{sec:prog-log} and \ref{sec:weakest-pre-rules} above to demonstrate how either type of rule can be used. We also use the consequence rule along with the program rule to be applied in several places in order to simplify presentation. As a result, some detailed analysis is omitted, such as proving \spps{} are disjoint in order to use the frame rule.
\allowdisplaybreaks
\begin{align*}
    &\{ \lst(i) \land \lst(j) \land \Fr(\lst(i)) \cap \Fr(\lst(j)) = \emptyset \land i \neq nil \} \tag{consequence rule} \\
    &\{ \lst(i) \land \lst(j) \land \Fr(\lst(i)) \cap \Fr(\lst(j)) = \emptyset \land i \neq nil \land i \notin \Fr(\lst(j))\} \tag{consequence rule: unfolding list definition} \\
    \vspace*{2ex}
    &\{\exists k' : k' = \nxt(i).\; \lst(k') \land i \not\in \Fr(\lst(k')) \land \lst(j) \\ &\hspace*{10ex} \land i \not\in \Fr(\lst(j)) \land \Fr(\lst(k')) \cap \Fr(\lst(j)) = \emptyset\} \tag{consequence rule} \\
    &\{\exists k' : k' = \nxt(i).\; \nxt(i) = \nxt(i) \land \lst(k') \land i \not\in \Fr(\lst(k')) \land \lst(j) \\ &\hspace*{10ex} \land i \not\in \Fr(\lst(j)) \land \Fr(\lst(k')) \cap \Fr(\lst(j)) = \emptyset\} \\
    &\hspace*{5ex} \texttt{k := i.next ;} \tag{consequence rule, lookup-G rule} \\
    &\{ \nxt(i) = \nxt(i) \land \lst(k) \land i \not\in \Fr(\lst(k)) \land \lst(j) \\ &\hspace*{10ex} \land i \not\in \Fr(\lst(j)) \land \Fr(\lst(k)) \cap \Fr(\lst(j)) = \emptyset \} \\
    &\hspace*{5ex} \texttt{i.next := j ;} \tag{mutation rule, frame rule} \\
    &\{ \nxt(i) = j \land \lst(k) \land i \not\in \Fr(\lst(k)) \land \lst(j) \\ &\hspace*{10ex} \land i \not\in \Fr(\lst(j)) \land \Fr(\lst(k)) \cap \Fr(\lst(j)) = \emptyset \} \tag{consequence rule} \\
    &\{ \lst(k) \land \nxt(i) = j \land i \not\in \Fr(\lst(j)) \land \lst(j) \land \Fr(\lst(k)) \cap \Fr(\lst(j)) = \emptyset \} \tag{consequence rule: folding list definition} \\
    &\{ \lst(k) \land \lst(i) \land \Fr(\lst(k)) \cap \Fr(\lst(i)) = \emptyset \} \\
    &\hspace*{5ex} \texttt{j := i ; i := k} \tag{assignment-G rule} \\
    &\{ \lst(i) \land \lst(j) \land \Fr(\lst(i)) \cap \Fr(\lst(j)) = \emptyset \}
\end{align*}
Armed with this, proving $j$ is a list after executing the full program above is a trivial application of the assignment, while, and consequence rules, which we omit for brevity.

Observe that in the above proof we were applying the frame rule because $i$ belongs neither to $\Fr(\lst(k))$ nor $\Fr(\lst(j))$. This can be dispensed with easily using reasoning about first-order formulas with least-fixpoint definitions, techniques for which are discussed in Section~\ref{sec:discussion}.

Also note the invariant of the loop is precisely the intended meaning of $\lst(i) * \lst(j)$ in separation logic. In fact, as we will see in Section~\ref{sec:discussion}, we can define a \emph{first-order} macro $\Star$ as $\Star(\varphi, \psi) = \varphi \land \psi \land \Fr(\varphi) \cap \Fr(\psi) = \emptyset$.
We can use this macro to represent disjoint supports in similar proofs.

These proofs demonstrate what proofs of actual programs look like in our program logic. They also show that frame logic and our program logic can prove many results similarly to traditional separation logic. And, by using the derived operator $\Star$, very little even in terms of verbosity is sacrificed in gaining the flexibility of Frame Logic(please see Section~\ref{sec:discussion} for a broader discussion of the ways in which Frame Logic differs from Separation Logic and in certain situations offers many advantages in stating and reasoning with specifications/invariants).

\section{Expressing a Precise Separation Logic}\label{sec:seplog-translation}
In this section, we show that FL is expressive
by capturing a fragment of separation logic in frame logic; the fragment is a syntactic fragment of separation logic that defines 
only \emph{precise formulas}--- formulas 
that can be satisfied in at most one heaplet
for any store.
The translation also shows that frame logic can naturally and compactly capture such separation logic formulas. 

\subsection{A Precise Separation Logic}
As discussed in Section~\ref{sec:intro}, a crucial difference between separation logic and frame logic is that formulas in frame logic have uniquely determined \spps{}/heaplets, while this is not true in separation logic (the heaplet for $\alpha \lor \beta$ can be one that supports the truth of $\alpha$ or one that supports the truth of $\beta$ in Separation Logic). 
However, it is well known that in verification,
determined heaplets are very natural (most uses of separation logic in fact are precise) and sometimes desirable. For instance, see~\cite{brookes07} where precision is used crucially in some proof rules regarding resource sharing to give sound semantics to concurrent separation logic and~\cite{ohearn04} where precise formulas are proposed in verifying modular programs as imprecision causes ambiguity in function contracts in the frame rule. The work in~\cite{gotsman11} also relies on precise predicates in developing a Hoare logic for concurrent programs. Specifically, the conjunction rule is unsound in the context of imprecise predicates.

We define a fragment of separation logic that defines precise formulas (more accurately, we handle a slightly larger class inductively: formulas that when satisfiable have unique minimal heaplets for any given store). 

\begin{definition} \mh Fragment:
\label{def:mhfragment}
\begin{itemize}
\item \revII{$\mathit{sf}$: formulas over the stack only (nothing dereferenced). Includes $\mathit{isatom}?()$, $m(x) = y$ for immutable $m$, $\top$, background formulas, and boolean combinations of these formulas}.
\item $x \xrightarrow{f} y$
\item $\mathit{ite}(\mathit{sf}, \varphi_1, \varphi_2)$ where $\mathit{sf}$ is from the first bullet
\item $\varphi_1 \land \varphi_2$ and
$\varphi_1 * \varphi_2$
\item $\mathcal I$ where $\mathcal I$ contains all unary inductive definitions $I$ that have unique heaplets inductively ($\mathit{list}, \mathit{tree}$, etc.). In particular, the body $\rho_I$ of $I$ is a formula in the \mh fragment ($\rho_I[I \mapsfrom \varphi]$ is in the \mh fragment provided $\varphi$ is in the \mh fragment). Additionally, for all $x$, if $s, h \models I(x)$ and $s, h' \models I(x)$, then $h = h'$.\footnote{While we only assume unary inductive definitions here, we can easily generalize this to inductive definitions with multiple parameters.}
\item $\exists y.\; (x \xrightarrow{f} y) * \varphi_1$
\end{itemize}
\end{definition}

Note that in the fragment negation and disjunction are disallowed, but mutually exclusive disjunction using $\ite$ is allowed. Existential
quantification is only present when the topmost operator is a $*$ and where one of the formulas guards the quantified variable uniquely.

The semantics of this fragment follows the standard semantics of separation logic~\cite{demri15,ohearn12,ohearn01,reynolds02}, with the heaplet of $x \xrightarrow{f} y$ taken to be 
$\{ x \}$. 
See Remark~\ref{rem:mutation-frame} in Section~\ref{sec:design} for a discussion of a more accurate heaplet for $x \xrightarrow{f} y$ being the
set containing the pair $(x,f)$, and how this can be modeled in the above semantics by using field-lookups using non-mutable pointers.

%

\begin{theorem}[Minimum Heap]
\label{minimalheap}
For any formula $\varphi$ in the \mh fragment, if there is an $s$ and $h$ such that $s, h \models \varphi$ then there is a $h_\varphi$ such that $s, h_\varphi \models \varphi$ and for all $h'$ such that $s, h' \models \varphi$, $h_\varphi \subseteq h'$.
\end{theorem}

\subsection{Translation to Frame Logic}
For a separation logic store and heap $s, h$ (respectively), we define the corresponding interpretation $\mathcal{M}_{s,h}$ such that variables are interpreted according to $s$ and values of pointer functions on $\mathit{dom}(h)$ are interpreted according to $h$. 
For $\varphi$ in the \mh fragment, we first define a formula $P(\varphi)$, inductively, that captures whether $\varphi$ is precise. $\varphi$ is a precise formula iff, when it is satisfiable with a store $s$, there is exactly one $h$ such that $s, h \models \varphi$.
The formula $P(\varphi)$ is in separation logic
and will
be used in the translation. 
To see why this formula is needed, consider the formula $\varphi_1 \wedge \ite(\mathit{sf}, \varphi_2, \varphi_3)$.
Assume that $\varphi_1$ is imprecise, $\varphi_2$ is precise, and $\varphi_3$ is imprecise. Under conditions where $\mathit{sf}$ is true, the heaplets for $\varphi_1$ and $\varphi_2$ must align. However, when $\mathit{sf}$ is false, the heaplets for $\varphi_1$ and $\varphi_3$ can be anything. Because we cannot initially know when $\mathit{sf}$ will be true or false, we need this separation logic formula $P(\varphi)$ that is true exactly when $\varphi$ is precise.

\begin{definition} Precision predicate $P$:
\label{def:precision}
\begin{itemize}
\item $P(\mathit{sf}) = \bot$ and
$P(x \xrightarrow{f} y) = \top$
\item $P(\mathit{ite}(\mathit{sf}, \varphi_1, \varphi_2)) = (\mathit{sf} \wedge P(\varphi_1)) \vee (\neg\mathit{sf} \wedge P(\varphi_2))$
\item $P(\varphi_1 \wedge \varphi_2) = P(\varphi_1) \vee P(\varphi_2)$
\item $P(\varphi_1 * \varphi_2) = P(\varphi_1) \wedge P(\varphi_2)$
\item $P(I) = \top$ where $I \in \mathcal{I}$ is an inductive predicate
\item $P(\exists y.\; (x \xrightarrow{f} y) * \varphi_1) = P(\varphi_1)$
\end{itemize}
\end{definition}

Note that this definition captures precision within our fragment since stack formulas are imprecise and pointer formulas are precise. The argument for the rest of the cases follows by simple structural induction. \revII{The interesting case is that of an inductive predicate $I$ whose body is assumed to be in the \mh fragment. We can show this by induction on the rank of $I(x)$\footnote{\revII{By the Knaster-Tarski theorem, we can compute the least-fixpoint by computing an increasing sequence pre-fixpoints for the kinds of definitions we assume. This yields a `rank' for every $x$ such that $I(x)$ holds, namely the index in the sequence of pre-fixpoints when $I(x)$ first holds.}} for an argument $x$ when $I(x)$ holds.}

Now we define the translation $T$ inductively:

\begin{definition} Translation from \mh to Frame Logic:
\label{def:translation}

\begin{itemize}
\item $T(\mathit{sf}) = \mathit{sf}$ and
$T(x \xrightarrow{f} y) = (f(x) = y)$
\item $\mathit{ite}(\mathit{sf}, \varphi_1, \varphi_2) = \mathit{ite}(T(\mathit{sf}), T(\varphi_1), T(\varphi_2))$
\item $T(\varphi_1 \wedge \varphi_2) =$
$\begin{array}[t]{rcl}
T(\varphi_1) \wedge T(\varphi_2) &\wedge\;  T(P(\varphi_1)) \implies \Fr(T(\varphi_2)) \subseteq \Fr(T(\varphi_1))\\
&\wedge\; T(P(\varphi_2)) \implies \Fr(T(\varphi_1)) \subseteq \Fr(T(\varphi_2))
\end{array}$
\item $T(\varphi_1 * \varphi_2) = T(\varphi_1) \wedge T(\varphi_2) \wedge \Fr(T(\varphi_1)) \cap \Fr(T(\varphi_2)) = \emptyset$
\item $T(I) = T(\rho_I)$ where $\rho_I$ is the definition of the inductive predicate $I$ as in Section 3.
\item $T(\exists y.\; (x \xrightarrow{f} y) * \varphi_1) = \exists y: [f(x) = y].\; [T(\varphi_1) \wedge x \not\in \Fr(T(\varphi_1))]$
\end{itemize}
\end{definition}

Finally, recall that any formula $\varphi$ in the \mh fragment has a unique minimal heap (Theorem~\ref{minimalheap}). With this (and a few auxiliary lemmas that can be found in Section~\ref{sec:appendix-seplogiclemmas}), we have the following theorem, which captures the correctness of the translation:


\begin{theorem}\label{seplogicmain}
For any formula $\varphi$ in the \mh fragment, we have the following implications:
\[ \begin{array}[t]{rcl}
s, h \models \varphi &\implies& \mathcal{M}_{s,h} \models T(\varphi) \\
\mathcal{M}_{s,h} \models T(\varphi) &\implies& s, h' \models \varphi \text{  where $h' \equiv \mathcal{M}_{s,h}(\Fr(T(\varphi)))$}
\end{array}
\]

\noindent Here, $\mathcal{M}_{s,h}(\Fr(T(\varphi)))$ is the interpretation of $\Fr(T(\varphi))$ in the model $\mathcal{M}_{s,h}$.
Note $h'$ is minimal and is equal to $h_\varphi$ as in Theorem~\ref{minimalheap}.
\end{theorem}

\subsection{Frame Logic Can Capture the \mh fragment: Proofs}
\label{sec:appendix-seplogiclemmas}

\begin{lemma}\label{extensibleheap}
For any formula $\varphi$ in the \mh fragment, if there is an $s$ and $h$ such that $s, h \models \varphi$ and we can extend $h$ by some nonempty $h'$ such that $s, h \cup h' \models \varphi$, then for any $h''$, $s, h \cup h'' \models \varphi$.
\end{lemma}
\begin{proof}
If a stack formula holds then it holds on any heap. Pointer formulas and inductive definitions as defined can never have an extensible heap so this is vacuously true.

For $\mathit{ite}(\mathit{sf}, \varphi_1, \varphi_2)$, assume WLOG $s, h \models \mathit{sf}$. Then for any $h'$, $s, h' \models \varphi_1 \Leftrightarrow s, h' \models \mathit{ite}(\mathit{sf}, \varphi_1, \varphi_2)$. Then use the induction hypothesis.

For $\varphi_1 \wedge \varphi_2$, for any $h'$, $s, h' \models \varphi_1 \wedge \varphi_2 \Leftrightarrow s, h' \models \varphi_1$ and $s, h' \models \varphi_2$. If the conjoined formula can be extended, both subformulas can be extended, and then we apply the induction hypothesis.

For separating conjunction, the nature of the proof is similar to conjunction, noting that the heap can be extended iff the heap of \textit{either} subformula can be extended.

For existential formulas in our form, the proof is again similar, noting the heap is extensible iff the heap of $\varphi_1$ is extensible.
\end{proof}

\begin{reptheorem}{minimalheap}
For any formula $\varphi$ in the \mh fragment, if there is an $s$ and $h$ such that $s, h \models \varphi$ then there is an $h_\varphi$ such that $s, h_\varphi \models \varphi$ and for all $h'$ such that $s, h' \models \varphi$, $h_\varphi \subseteq h'$.
\end{reptheorem}
\begin{proof}
The minimal heaplets for stack formulas are empty. For $x \xrightarrow{f} y$ the heaplet is uniquely $\{x\}$. 

For conjunction, there are three cases depending on if $\varphi_1$ or $\varphi_2$ or both have extensible heaplets. We cover the most difficult case where they both have extensible heaplets here. By definition we know $s, h \models \varphi_1$ and $s, h \models \varphi_2$. By induction, we know there are unique $h_{\varphi_1}$ and $h_{\varphi_2}$ such that $h_{\varphi_1}$ and $h_{\varphi_2}$ model $\varphi_1$ and $\varphi_2$ respectively and are minimal. Thus, $h_{\varphi_1} \subseteq h$ and $h_{\varphi_2} \subseteq h$, so $h_{\varphi_1} \cup h_{\varphi_2} \subseteq h$. By Lemma~\ref{extensibleheap}, $h_{\varphi_1} \cup h_{\varphi_2}$ is a valid heap for both $\varphi_1$ and $\varphi_2$. Thus, $s, h_{\varphi_1} \cup h_{\varphi_2} \models \varphi_1 \wedge \varphi_2$ and $h_{\varphi_1} \cup h_{\varphi_2}$ is minimal.

For separating conjunction the minimal heaplet is (disjoint) union. For $\mathit{ite}$ we pick the heaplet of either case depending on the truth of the guard. By definition, inductive defintions will have minimal heaplets.

Inductive definitions have unique heaplets by the choice we made above and therefore vacuously satisfy the given statement.

For existentials, we know from the semantics of separation logic that every valid heap on a store $s$ for the original existential formula is a valid heap for $\psi \equiv (x \xrightarrow{f} y) * \varphi_1$ on a modified store $s' \equiv s[y \mapsto v]$ for some $v$. Since the constraint $(x \xrightarrow{f} y)$ forces the value $v$ to be unique, we can then invoke the induction hypothesis to conclude that the minimal heaplets of the existential formula on $s$ and of $\psi$ on $s'$ are the same. In particular, this means that existential formulas in our fragment also have a minimal heaplet.
\end{proof}

\begin{lemma}\label{minimalframeheap}
For any $s, h$ such that $s, h \models \varphi$ we have $\mathcal{M}_{s,h}(\Fr(T(\varphi))) = h_\varphi$ where $h_\varphi$ is as above.
\end{lemma}
\begin{proof}
Structural induction on $\varphi$.

If $\varphi$ is a stack formula, $h_\varphi = \Fr(T(\varphi)) = \emptyset$. If $\varphi \equiv x \xrightarrow{f} y$, $h_\varphi = \Fr(T(\varphi)) = \{x\}$.

For $\varphi \equiv \mathit{ite}(\mathit{sf}, \varphi_1, \varphi_2)$, because, $s, h \models \varphi$, we know either $s, h \models \varphi_1$ or $s, h \models \varphi_2$ depending on the truth of $\mathit{sf}$. WLOG assume $s, h \models \mathit{sf}$, then $h_\varphi = h_{\varphi_1}$. Similarly, $\Fr(T(\varphi)) = \Fr(\mathit{sf}) \cup \Fr(T(\varphi_1)) = \Fr(T(\varphi_1))$ (heaplet of stack formulas is empty) and then we apply the induction hypothesis. Similarly if $s, h \not\models \mathit{sf}$.

For $\varphi \equiv \varphi_1 \wedge \varphi_2$, we know from the proof of Theorem~\ref{minimalheap} that $h_\varphi = h_{\varphi_1} \cup h_{\varphi_2} = \mathcal{M}_{s,h}(\Fr(T(\varphi_1))) \cup \mathcal{M}_{s,h}(\Fr(T(\varphi_2)))$. The guard parts of the translation $\Fr(\varphi)$ since they are all precise formulas which have empty heaplets.

For $\varphi \equiv \varphi_1 * \varphi_2$, the proof is the same to the previous case, again from the proof of Theorem~\ref{minimalheap}.


For an inductive definition $I$, recall that $\rho_I[I \mapsfrom \varphi]$ is in the \mh fragment (and crucially does not mention $I$). Assume $\varphi$ is fresh and does not occur in $\rho_I$. Define $\rho_I' \equiv \rho_I[I \mapsfrom \varphi]$ and note that $\rho_I = \rho_I'[\varphi \mapsfrom I]$. This means that $h_{\rho_I} = h_{\rho_I'}[h_{\varphi} \mapsfrom h_I]$. We also see that $\Fr(T(\rho_I)) = \Fr(T(\rho_I'))[\Fr(T(\varphi)) \mapsfrom \Fr(T(\rho_I))]$. Because $h_{\rho_I'} = \Fr(T(\rho_I'))$ (by the other cases in this proof and since $\rho_I'$ does not mention $I$), we see the heaplets are related by the same sets of recursive equations and we are done.

For existentials, we have from the definition of the $\Fr$ operator that the \spp{} of the translation of the existential formula is the same as that of $\{x\} \cup \Fr(T(\varphi_1))$. The claim then follows from the definition of heaplet of existentials in separation logic as well as the inductive hypothesis for $\varphi_1$. 
\end{proof}


\begin{reptheorem}{seplogicmain}
For any formula $\varphi$ in the \mh fragment, we have the following implications:
\begin{align*}
s, h \models \varphi &\implies \mathcal{M}_{s,h} \models T(\varphi) \\
\mathcal{M}_{s,h} \models T(\varphi) &\implies s, h' \models \varphi \text{  where $h' \equiv \mathcal{M}_{s,h}(\Fr(T(\varphi)))$}
\end{align*}
Here, $\mathcal{M}_{s,h}(\Fr(T(\varphi)))$ is the interpretation of $\Fr(T(\varphi))$ in the model $\mathcal{M}_{s,h}$.
Note $h'$ is minimal and is equal to $h_\varphi$ as in Theorem~\ref{minimalheap}.
\end{reptheorem}
\begin{proof}
First implication: Structural induction on $\varphi$. 

If $\varphi$ is a stack formula or a pointer formula, this is true by construction. If $\varphi$ is an if-then-else formula the claim is true by construction and the induction hypothesis.

If $\varphi = \varphi_1 \land \varphi_2$, we know by the induction hypothesis that $\mathcal{M}_{s,h} \models T(\varphi_1)$ and $\mathcal{M}_{s,h} \models T(\varphi_2)$. Further, from the semantics of separation logic, we have that if $\varphi_1$ is precise, then $h_{\varphi_1} = h$. Therefore, $h_{\varphi_2} \subseteq h_{\varphi_1}$ (by Lemma~\ref{minimalheap}). Therefore, from Lemma~\ref{minimalframeheap}, we have that $\mathcal{M}_{s,h} \models \Fr(T(\varphi_2)) \subseteq \Fr(T(\varphi_1))$. Similarly if $\varphi_2$ is precise. This justifies the two latter conjuncts of the translation.

If $\varphi = \varphi_1 * \varphi_2$, we know there exist $h_1, h_2$ such that $h_1 \cap h_2 = \emptyset$ and $s, h_1 \models \varphi_1$ and $s, h_2 \models \varphi_2$. Then, from Lemma~\ref{minimalheap}, we have that $h_{\varphi_1} \subseteq h_1$ and $h_{\varphi_2} \subseteq h_2$. Thus, by Lemma~\ref{minimalframeheap}, we have that $\mathcal{M}_{s,h} \models \Fr(T(\varphi_1)) \cap \Fr(T(\varphi_2)) = \emptyset$. The other conjuncts follow from the induction hypothesis.

Similarly to the proof of Lemma~\ref{minimalframeheap}, we can show that the translation of the inductive definition satisfies the same recursive equations as the original inductive definition and we are done.

If $\varphi$ is an existential, the result follows from definition and the induction hypothesis.

Second implication: Structural induction on $\varphi$.

By construction, induction hypotheses, and Lemma~\ref{minimalframeheap}, all cases can be discharged besides conjunction and inductive predicates.

For conjunction, if $\varphi = \varphi_1 \land \varphi_2$, we have from the induction hypothesis that $s, h_{\varphi_1} \models \varphi_1$ and $s, h_{\varphi_2} \models \varphi_2$. If $\varphi_1$ is precise, we know $\mathcal{M}_{s,h} \models \Fr(T(\varphi_2)) \subseteq \Fr(T(\varphi_1))$ and therefore $h_{\varphi_2} \subseteq h_{\varphi_1}$ (from Lemma~\ref{minimalframeheap}). Similarly, if $\varphi_2$ is precise, then $\mathcal{M}_{s,h} \models \Fr(T(\varphi_1)) \subseteq \Fr(T(\varphi_2))$ as well as $h_{\varphi_1} \subseteq h_{\varphi_2}$. In particular, if they are both precise, their \spps{} (and therefore minimal heaplets) are equal, and $h' = h_{\varphi_1} \cup h_{\varphi_2}$ (from the proof of Lemma~\ref{minimalheap}) $= h_{\varphi_1} = h_{\varphi_2}$, and we are done. If only $\varphi_1$ is precise (similarly if only $\varphi_2$ is precise), then we have as above that $h_{\varphi_2} \subseteq h_{\varphi_1}$ and $h_{\varphi_1} = h'$. Moreover, we know by Lemma~\ref{extensibleheap} that $s, h_{\varphi_1} \models \varphi_2$ and we are done. If neither is precise, both heaps are extensible, so we know by Lemma~\ref{extensibleheap} that $s, h_{\varphi_1} \cup h_{\varphi_2} \models \varphi_1$ and $s, h_{\varphi_1} \cup h_{\varphi_2} \models \varphi_2$ and we are done.

For $\varphi$ an inductive predicate, we know that $\mathcal{M}_{s,h} \big | \Fr(T(\varphi)) \models T(\varphi)$. The remainder follows since, because we restrict the form of inductive predicates to have a unique heap at each level, the translated inductive predicate will satisfy the same recursive equations as $\varphi$.

\end{proof}

\section{Discussion}\label{sec:discussion}

\subsection{Comparison with Separation Logic}
The design of frame logic is, in many ways, inspired by the design choices of separation logic. Separation logic formulas implicitly hold on \emph{tight} heaplets--- models are defined on pairs $(s,h)$, where $s$ is a store (an interpretation of variables) and $h$ is a heaplet that defines a subset of the heap as the domain for functions/pointers. 
In Frame Logic, we choose to not define satisfiability with respect to heaplets but define it with respect to the entire global heap. However, we give access to the implicitly defined heaplet using the operator $\Fr$, and give a logic
over \emph{sets} to talk about supports. 
The separating conjunction operation $*$ can then be expressed using normal conjunction of the two formulae and a constraint that says that the supports of the formulas are disjoint \revII{(e.g., $x \overset{\mathit{next}}{\mapsto} y * \mathit{list}(y)$ can be expressed in FL as $\mathit{next}(x) = y \land \mathit{list}(y) \land \Fr(\mathit{next}(x) = y) \cap \Fr(\mathit{list}(y)) = \emptyset$).
}

We do not allow formulas to have \emph{multiple} supports which is crucial ($\Fr$ is a function). 
Precise fragments of separation logic have been proposed and accepted in the separation logic literature a way of giving robust semantics in handling modular functions and concurrency~\cite{ohearn04,brookes07}. 
Section~\ref{sec:seplog-translation} details a translation of a precise fragment of separation logic (with $*$ but not magic wand) to frame logic that shows the natural connection between precise formulas in separation logic and frame logic. When converting arbitrary (potentially imprecise) separation logic formulas to first-order logic, there is an inherent existential quantifier over heaplets that makes reasoning difficult, due to the non-uniqueness of heaplets. Our logic however does not have this problem because supports are uniquely defined, and hence we do not need to quantify over heaplets (as in the above example).

Frame logic, through the support operator, facilitates local reasoning much in the same way as separation logic does, and the frame rule in frame logic supports frame reasoning in a similar way as the frame rule in separation logic. The key difference between frame logic and separation logic is the adherence to a first-order logic (with recursive definitions), both in terms of syntax and expressiveness. Note that there is no difference in regards to completeness, as \revII{both Separation Logic and FO-RD are incomplete~\cite{reynolds02,sl-computability-complexity,loding18}}. We refer the reader to page~\pageref{sec:consequences} where several nuances in semantic choices are explained. 

There are several other key differences between separation logic and frame logic. First, in separation logic, the magic wand operator is needed to express the weakest precondition~\cite{reynolds02}. Consider for example computing the weakest precondition of the formula $list(x)$ with respect to the code
$y.n := z$. The weakest precondition should essentially describe the (tight) heaplets such that changing the $n$ pointer from $y$ to $z$ results in $x$ pointing to a list. In separation logic, this is expressed typically (see~\cite{reynolds02}) using magic wand as
$(y \xrightarrow{n} z) ~-\!\!*~ (list(x))$. 
\revII{However, the magic wand operator is inherently a \emph{second-order} property~\cite{magic-wand-complexity}.} The formula $\alpha -\!\!* \beta$ holds on a heaplet $h$ if for any \emph{disjoint} heaplet that satisfies $\alpha$, $\beta$ will hold on the conjoined heaplet. Expressing this property (for arbitrary $\alpha$, whose heaplet can be \emph{unbounded}) requires quantifying over unbounded heaplets satisfying $\alpha$, which is not first order
expressible. 

In frame logic, we instead rewrite the recursive definition $list(\cdot)$ to a new one $list'(\cdot)$ that captures whether $x$ points to a list, assuming that $n(y)=z$ (see~Section~\ref{sec:weakest-pre-rules}). This property continues to be expressible in frame logic and can be converted to first-order logic with recursive definitions (see~Section~\ref{sec:reduction-to-FO-RD}). Note that we are exploiting the fact that there is only a bounded amount of change to the heap in loop-free programs in order to express this in FL.

Let us turn to expressiveness and succinctness. In separation logic, separation of structures is expressed using $*$, and in frame logic, such a separation is expressed using conjunction and an additional constraint that says that the supports of the two formulas are disjoint. A precise separation logic formula of the form $\alpha_1 * \alpha_2 * \ldots \alpha_n$ is succinct and would get translated to a much larger formula in frame logic as it would have to state that the supports of each pair of formulas is disjoint.
We believe this can be tamed using macros ($\textit{Star}(\alpha, \beta) = \alpha \wedge \beta \wedge \Fr(\alpha) \cap \Fr(\beta)=\emptyset$). 

There are, however, several situations where frame logic leads to more compact and natural formulations. For instance, consider expressing the property that $x$ and $y$ point to lists, which may or may not overlap. 
In Frame Logic, we simply write $list(x) \land list(y)$. The support of this formula is the union of the supports of the two lists. 
In separation logic, we cannot use $*$ to write this compactly (while capturing the tightest heaplet). 
Note that the formula $(list(x) * true) \wedge (list(y) * true)$ is \emph{not} equivalent, as it is true in heaplets that are larger than the set of locations of the two lists.
The simplest formulation we know is to write a recursive definition $\mathit{lseg}(u,v)$ for list segments from $u$ to $v$ and use quantification: $(\exists z.~\mathit{lseg}(x,z)*\mathit{lseg}(y,z)*list(z))
 \vee (list(x)*list(y))$
where the definition of $\mathit{lseg}$ is the following: $\mathit{lseg}(u,v) \equiv (u = v \land \mathit{emp}) \lor (\exists w.~ u \rightarrow w ~*~ lseg(w,v))$.

If we wanted to say $x_1,\ldots,x_n$ all point to lists, that may or may not overlap, then in FL we can say $list(x_1) \land list(x_2) \land \ldots \land list(x_n)$.
However, in separation logic, the simplest way seems to be to write using $\mathit{lseg}$ and a linear number of quantified variables and an exponentially-sized formula. Now consider the property saying $x_1,\ldots,x_n$ all point to binary trees, with
pointers $\textit{left}$ and $\textit{right}$, and that can overlap arbitrarily. We can write it in FL as $tree(x_1)\land\ldots\land tree(x_n)$, while a formula in separation logic that expresses this property seems very complex.

\revII{The difficulty of using separation logic to capture such overlapping datastructures~\cite{iterated-star-sidkrish2020,hongseok-overlaid-cav2011} has been noticed in the past. For example, the work in~\cite{hobor13} introduces an overlapping conjunction operator $\vcup{*}$ to separation logic, which unlike $*$ or $\land$, allows the conjunction of two properties whose heaplets may overlap but not coincide, and can express the above properties we mention succinctly. However, the \emph{ramification} rules proposed in that paper to reason with overlapping datastructures seem inherently second-order, introducing the magic wand~\cite{magic-wand-complexity}. The work in~\cite{guarded-wand-pagel2020} explores a decidable fragment of separation logic with inductive definitions including \emph{guarded} magic wands to capture overlaid datastructures.}

\revII{In summary, we believe that frame logic is a logic that supports frame reasoning built on similar principles as separation logic, but is still translatable to first-order logic (avoiding the magic wand) while still being closed under weakest preconditions. Furthermore, the choices
it makes in syntax and semantics lead to expressing certain properties more naturally and compactly, while others more verbosely.
}

\subsection{\revII{Reasoning with Frame Logic}}
\label{sec:reasoning-with-fl}
\revII{A important advantage of the adherence of frame logic is that it is translatable to a first-order logic with recursive definitions, opening up the possibility of using solvers for the latter in order to reason with frame logic. }


\revII{Notably, we are inspired by the reasoning using the framework on Natural Proofs~\cite{pek14,qiu13,loding18}.
In this work, recursive definitions with least fixpoint semantics are first abstracted to fixpoint definitions, lending to a formulation in pure first-order logic. Second, universal quantifiers are instantiated using terms in the formula, and recursive definitions are unfolded on terms as well. Finally, the resulting formulae, which are quantifier-free are checked using SMT solvers. A remarkable result is that for a \emph{safe} fragment of the logic (which verification conditions for datastructure manipulating programs often adhere to), the above technique is \emph{complete} with respect to pure FO reasoning~\cite{loding18}. The work in~\cite{fossil} also attempts to bridge the gap between the first-order reasoning of natural proofs and the least-fixpoint semantics of the inductive definitions by synthesizing inductive lemmas.}

\revII{Our main hope of reasoning with frame logic resides in this technique, which has already been used in reasoning with {\sc Dryad}, a precise fragment of separation logic, by converting it to first-order logic with recursive definitions~\cite{pek14,qiu13}.
The evaluation of \textsc{Dryad} on a suite of 150 heap programs has been reported in the works on Natural Proofs (see http://madhu.cs.illinois.edu/vcdryad/) providing evidence that the technique works well
in practice.
However, we believe that verification conditions using the weakest precondition derived in this paper introduce certain quantifications that may make it harder than necessary to reason with. 
We believe an approach of verification condition generation based on strongest postconditions (or symbolic execution), as done in~\cite{qiu13,pek14}, and with the frame logic restricted in order to not introduce existential quantification in specifications, would lead to an efficient solver.
In particular, the support of recursively defined functions can be generated automatically using 
our translation, while the work in~\cite{pek14,qiu13} translated heaplets for {\sc Dryad} formulas using manual unverified translations. 
This work is however beyond the scope of this paper, and we leave it for the future. The main technical challenge is to restrict frame logic annotations to not use existential quantification and realize supports of inductively defined datastructures as FO inductive definitions over set without introducing additional quantifiers.}

\revII{Another mechanism to reason with frame logic is to convert it into FO-RD, and use it in frameworks such as {\sc Dafny}~\cite{dafny} or {\sc Boogie}~\cite{boogie,boogie-meets-regions} that support rich user-annotation based reasoning. In particular, the methodology of \emph{region logics}~\cite{RegionLogic1,RegionLogic2,RegionLogic} suggests ways of encoding supports as sets in first-order logics, and carry through proofs using \emph{ghost annotations} that update a ghost state and help prove theorems. This methodology typically requires help from the user in terms of ghost annotations, but much of the reasoning otherwise is delegated to automatic logic engines (mainly SMT solvers). Region logic~\cite{RegionLogic1,RegionLogic2,RegionLogic} itself is a logic for heap manipulating programs that supports \emph{explicit} heaplets and frame-based reasoning using them, rather than the implicit heaplets of frame logic.}

\section{Related Work}\label{sec:rel-work}
The frame problem~\cite{hayes81, borgida95} is an important problem in many different domains of research. In the broadest form, it concerns representing and reasoning about the effects of a local action without requiring explicit reasoning regarding static changes to the global scope. For example, in artificial intelligence one wants a logic that can seamlessly state that if a door is opened in a lit room, the lights continue to stay switched on. This issue is present in the domain of verification as well, specifically with heap-manipulating programs.

There are many solutions that have been proposed to this problem. The most prominent proposal in the verification context is  separation logic~\cite{demri15,ohearn12,ohearn01,reynolds02}, which 
we discussed in detail in the previous section.



\revII{The work on Dynamic Frames~\cite{kassios06,DynamicFrames2011} and similarly inspired approaches such as Region Logic~\cite{RegionLogic1,RegionLogic2,RegionLogic} allow programs to explicitly specify parts of the heap that may be modified. 
The key idea is the notion of \emph{regions}, which are subsets of locations that can be manipulated in code (as ghost variables) as well as used in annotations in order to perform frame reasoning.  
This allows a finer-grained specification of the portion of the heap modified by programs, 
avoiding special symbols like $*$ and $-*$. Section~\ref{sec:reasoning-with-fl} discusses on verification methodologies based on these logics as well as using using these techniques to reasoning with FL. In contrast, FL itself has \emph{implicit} supports given using the $\Fr$ operator. 
}


The work on Implicit Dynamic Frames~\cite{smans12,chalice,idf,parkinson11} \revII{bridges the worlds of separation logic 
and dynamic frames}--- it uses separation logic and fractional permissions to implicitly define frames (reducing annotation burden), allows annotations to access these frames, and translates them into set regions for first-order reasoning. Our work is similar in that frame logic also implicitly defines regions and gives annotations access to these regions, and can be easily translated to pure FO-RD for first-order reasoning. However, implicit dynamic frames do not allow a formula to access anything in a function that is not in the support of the function's precondition. Further, all implementations of Implicit Dynamic Frames such as Chalice~\cite{parkinson11} and Viper~\cite{vipertool} only handle restricted formulas. \revII{For example, non-separating conjunction is prohibited.} 
In addition, a translation from the VeriFast thereom prover to Implicit Dynamic Frames~\cite{verifast,jost14} contains similarly restricted formulas.



One distinction with separation logic involves the non-unique heaplets in separation logic and the unique heaplets in frame logic. Determined heaplets have been used~\cite{qiu13,pek14,ohearn04} as they are more amenable to automated reasoning. In particular a separation logic fragment with determined heaplets known as precise predicates is defined in~\cite{ohearn04}, which we capture using frame logic in Section~\ref{sec:seplog-translation}.

There is also a rich literature on reasoning with these heap logics for program verification. Decidability is an important dimension and there is a lot of work on decidable logics for heaps with separation logic specifications~\cite{Smallfoot,BerdineCalcagnoOHearn2004,BerdineCalcagnoOHearn2005,CookHaaseChristoph2011,PerezAntonioRybalchenko2011,PerezRybalchenko2013}. The work based on EPR (Effectively Propositional Reasoning) for specifying heap properties~\cite{Itzhaky2014,Itzhaky2013,ItzhakyEPRInvariants} provides decidability, as does some of the work that translates separation logic specifications into classical logic~\cite{PiskacWiesZufferey2013}. 

\revII{The work in~\cite{bobot12} defines an extension of first-order logic and inductive definitions with footprints, where the design decisions on the semantics of footprints are very similar to ours. The work develops the notion of `separation predicates' based on these footprint expressions as well as verification conditions for programs annotated with first-order formulas involving separation predicates. The work also realizes an implementation of this technique by generating verification conditions and using Why3~\cite{filliatre13esop,why3logic}. The footprint expressions compute the same set as our $\Fr$ operator (with some minor syntactic changes), but our work handles a much larger fragment. In particular, apart from only defining supports for recursive predicates, the work in~\cite{bobot12} disallows quantification, non-separating conjunction, and separation predicates in the body of recursive predicates. This makes it difficult to define many of the typical data-structures shown in Figure~\ref{fig:rec-defs-list} in their logic.}
    

Translating separation logic into other logics and reasoning with them is another solution pursued in a lot of recent efforts~\cite{ChinDavidNguyen2007,pek14,PiskacWiesZufferey2013,PiskacWiesZufferey2014,PiskacWiesZufferey2014Tool,madhusudan12,pek14,qiu13,suter10,loding18}.
Other techniques including recent work on cyclic proofs~\cite{brotherston11,ta16} use heuristics for reasoning about recursive definitions. 

\revII{Finally, there is also work that studies first-order logics such as FO(ID) with more general non-monotonic inductive definitions motivated by AI applications such as the frame problem and the modeling of causal processes of knowledge~\cite{foid2000,foid-revisited2014}. FO(ID) generalizes the least-fixpoint semantics~\cite{tarski-knaster} we use for definitions in FO-RD (which are all required to be monotonic).} 

\section{Conclusions}\label{sec:conc}
Our main contribution is to propose \emph{Frame Logic}, a first-order logic endowed with an explicit operator
that recovers the implicit \spps{} of formulas and supports
frame reasoning. 
We have argued its expressiveness by capturing several properties of data-structures naturally and succinctly, and by showing that it can express a precise fragment of separation logic. 
The program logic built using frame logic supports local heap reasoning, frame reasoning, and weakest tightest preconditions across loop-free programs.

We believe that frame logic is an attractive alternative to separation logic, built using similar principles as separation logic while staying within the first-order logic world. The
first-order nature of the logic makes it potentially amenable to easier automated reasoning.


A practical realization of a tool for verifying programs in a standard programming language with frame logic annotations by marrying it with existing automated techniques and tools for first-order logic (in particular~\cite{madhusudan12,pek14,qiu13,suter10,kovacs17}), is the most compelling future work. 

Another area for future work involves potential extensions to Frame Logic. For example, the work in~\cite{krishnaswami09} uses magic wand for verifying certain design patterns including iterators, and extending Frame Logic similarly is an interesting possible direction. Other potential extensions include support for function pointers and permission models.

\smallskip



\bibliographystyle{ACM-Reference-Format}
\bibliography{refs}


\end{document}